\documentclass[11pt, lmargin=1.25in, rmargin=1.25in]{article}

\RequirePackage{amsthm,amsmath,amsfonts,amssymb}
\RequirePackage[authoryear]{natbib} 
\RequirePackage[colorlinks,citecolor=blue,urlcolor=blue,breaklinks]{hyperref}

\usepackage{bm}
\usepackage{mathrsfs}  
\usepackage{appendix}
\usepackage{commath}
\usepackage{dsfont}
\usepackage{subfig}
\usepackage{xcolor}
\usepackage{mathtools}
\usepackage{enumitem}
\usepackage{fancyhdr} 
\usepackage{geometry}
\usepackage{caption}
\usepackage{booktabs}
\usepackage{multirow}
\usepackage{lscape}

\usepackage{setspace}
\newcommand{\N}{\mathbb{N}}
\newcommand{\R}{\mathbb{R}}
\newcommand{\eps}{\varepsilon}
\renewcommand{\P}{\mathbb{P}}
\newcommand\E{\mathbb{E}}
\newcommand{\mc}{\mathcal}

\theoremstyle{plain}
\newtheorem{theorem}{Theorem}[section]
\newtheorem{lemma}[theorem]{Lemma}
\newtheorem{corollary}[theorem]{Corollary}
\newtheorem{proposition}[theorem]{Proposition}

\theoremstyle{remark}
\newtheorem{remark}{Remark}[section]
\newtheorem{example}{Example}[section]

\newtheorem{assumption}{Assumption}[section]

\begin{document}

\title{
Robust Instrumental Variables: Sharp Rates and Inference under Adversarial Contamination
}

\author{
\begin{tabular}{c}
Anders Bredahl Kock\footnote{
Kock's research was supported by the European Research Council (ERC) grant number 101124535 -- HIDI (UKRI EP/Z002222/1).  He is also a member of, and grateful for support from, i) the Aarhus Center for Econometrics (ACE), funded by the Danish National Research Foundation grant number DNRF186,  and ii) the Center for Research in Energy: Economics and Markets (CoRE).} \\ 
\footnotesize	University of Oxford \\
\footnotesize Department of Economics\\
\footnotesize	10 Manor Rd, Oxford OX1 3UQ
\\
\footnotesize	{\footnotesize	\href{mailto:anders.kock@economics.ox.ac.uk}{anders.kock@economics.ox.ac.uk}} 
\end{tabular}
\begin{tabular}{c}
David Preinerstorfer \\ 
{\footnotesize WU Vienna University of Economics and Business} \\
{\footnotesize Institute for Statistics and Mathematics} \\
{\footnotesize Welthandelsplatz 1, 1020 Vienna} \\ 
{\footnotesize	 \href{mailto:david.preinerstorfer@wu.ac.at}{david.preinerstorfer@wu.ac.at}}
\end{tabular}
}

\date{Preliminary version: July 2026}

\maketitle	
\begin{abstract}
Because 2SLS is built from sample averages, a small number of observations can have a disproportionate effect on estimates and inference. We introduce W-2SLS, a simple drop-in robustification that replaces these averages by quantile-winsorized means. We analyze W-2SLS under adversarial contamination, which permits both the identities and the reported values of the contaminated observations to depend on the realized clean sample and therefore accommodates targeted or strategic manipulation. Under finite $m$-th moments, W-2SLS attains the minimax-sharp rate~$\eta_{n}^{1-\frac1m}+n^{-1/2}$, where~$\eta_n$ is the fraction of observations that may be altered. Matching lower bounds identify the exact contamination thresholds for uniform consistency, root-$n$ estimation, and centered Gaussian inference with the same first-order law as clean-sample 2SLS. When~$\sqrt{n}\eta_{n}^{1-\frac1m}\to 0$ robustness is first-order free. We also construct feasible heteroskedasticity-robust inference and a winsorized Anderson--Rubin test valid under weak identification and adversarial contamination. Finally, even without contamination, ordinary 2SLS can have poor uniform finite-sample concentration, whereas W-2SLS admits confidence-calibrated sub-Gaussian deviation guarantees.
\end{abstract}

\section{Introduction}
Instrumental variables methods are among the central tools of empirical economics, but the usual two-stage least squares (2SLS) estimator is built from ordinary sample averages of the IV moments.  This makes 2SLS vulnerable in precisely the situations in which empirical IV designs are often used: samples with high leverage observations, heavy-tailed outcomes or instruments, or a small number of influential observations.  The practical relevance of this fragility has been documented in~\cite{young2022consistency},  (cf.~also~\cite{young2019channeling}) who re-analyzed 1{,}309 2SLS regressions published in the AER and other AEA journals. He found that deleting just two (out of sometimes thousands of) observations can turn 58\% of 0.05 significant 2SLS results insignificant. Similarly, 57\% of 0.05 insignificant results can be rendered significant by deleting just two observations or clusters. The changes can be extraordinary, with~$p$-values moving from close to zero to close to one, and vice-versa.

Furthermore, the contamination or manipulation of the data may be strategic: In economics, the ``data-generating process'' often passes through firms, schools, hospitals, banks, governments, customs agents, or respondents whose \emph{payoff depends on what they report}. Many papers have observed the prevalence of potentially data-dependent strategic manipulation. For example, \cite{angrist2017small} observed how standardized test scores in Southern Italy were manipulated and that smaller class sizes increased this manipulation; \cite{angrist2019maimonides} revealed enrolment manipulation near Maimonides' rule cut-offs used to determine class sizes in Israeli schools, implying that the instrument itself can be contaminated; \cite{camacho2011manipulation} and \cite{miller2013risk} documented strategic manipulation exactly at eligibility thresholds for welfare programs in Colombia which was particularly pronounced around local elections; \cite{duflo2013truth} showed how auditors paid by firms systematically reported the firm's emissions just below the regulatory standard, although true emissions were typically higher. Further evidence for (strategically) contaminated data can be found in \cite{jacob2003rotten}, \cite{fisman2004tax}, \cite{dafny2005hospitals}, \cite{dee2019causes}, \cite{geruso2020upcoding}, \cite{martinez2022much}, \cite{brodeur2016star}.\footnote{\cite{brodeur2016star} is an example of the (well-known) fact that researchers themselves may have an incentive to misreport or manipulate their results in response to, e.g., publication incentives.}

These examples show that the manipulation is often not random, but can depend on the clean data (e.g.~the true test scores, true income levels, or true amounts of pollution), which is \emph{unobservable} to the researcher. In addition, the manipulated data can be perfectly plausible, perhaps to disguise the manipulation, and therefore need not be a visually obvious outlier. Therefore, it is desirable to have statistical procedures that are simultaneously i) robust to a broad range of potentially strategic contaminations of the data and ii) provably precise under the weakest possible assumptions on the amount of contamination. 

In this paper we study an instrumental variable setting under \emph{adversarial contamination}. Under adversarial contamination, which is formally introduced Section~\ref{sec:AdvContam}, an adversary is allowed to inspect the clean data and change up to a fraction~$\eta_n$ of the observations, with~$n$ being the sample size.\footnote{As the adversary is a statistical device rather than a literal malicious actor, there need not be \emph{one} adversary contaminating the data. For example, there could be multiple school principals altering grades to achieve a threshold passing rate in each of their school's.} The identity of the observations that are altered and the values that these are changed to can (but need not) depend on the clean data, which only the adversary observes.\footnote{For example, the adversary may manipulate observations so that as many test scores as possible fall above a desired threshold. To decide which scores to alter, the adversary must clearly inspect the full set of true test scores implying that this contamination rule is a function of the clean data.} The adversary then returns the contaminated data to the researcher. Any feasible statistical procedure is a function of the contaminated data. The fact that adversarial contamination allows the adversary to perform data-dependent manipulation is crucial to handle the above examples of strategic manipulation and sets it apart from the classic Huber contamination model in robust statistics. Due to this flexibility, adversarial contamination has become a popular robustness framework in machine learning and mathematical statistics, cf.~e.g., \cite{lai2016agnostic}, \cite{cheng2019high}, \cite{diakonikolas2019robust}, \cite{hopkins2020robust}, \cite{LM21}, \cite{minsker2021robust}, \cite{bhatt2022minimax}, \cite{depersin2022robust}, \cite{dalalyan2022all}, \cite{minasyan2023statistically}, \cite{diakonikolas2023algorithmic}, \cite{minsker2023efficient}, \cite{oliveira2025finite}. 

To achieve robustness to adversarial contamination, we propose replacing the fragile sample averages in the classic 2SLS estimator by the winsorized means recently studied in \cite{Wins1} where the winsorization points are carefully chosen order statistics of the contaminated data. Thus, the resulting winsorized (W-2SLS) estimator is a simple robustification  of the classic one. In a setting wherein the adversary is allowed to arbitrarily alter up to a fraction~$\eta_n$ (equivalently up to a number~$q_n=\eta_nn$) of the observations, we show that the W-2SLS estimator taking the \emph{contaminated} data as input is suitably close to the classic 2SLS estimator taking the \emph{clean} data as input, the distance depending on~$\eta_n$ and the number of moments~$m$ that the clean data possesses. Based on this, we show that W-2SLS 
\begin{itemize}
	\item is consistent if~$\eta_n\to 0$. We show that no (sequence of) estimator(s) can be uniformly consistent if~$\eta_n\not\to 0$.
	\item converges at rate~$\sqrt{n}$ if~$\sqrt{n}\eta_n^{1-\frac1m}$ is bounded. We show that no estimator can uniformly converge at rate~$\sqrt{n}$ if $\sqrt{n}\eta_n^{1-\frac1m}$ is unbounded. 
	\item is asymptotically normal with a mean of zero and the same covariance matrix as the infeasible 2SLS estimator based on the unobserved clean data if~$\sqrt{n}\eta_n^{1-\frac1m}\to 0$. We show that no estimator can satisfy this if~$\sqrt{n}\eta_n^{1-\frac1m}\not\to 0$.
	\end{itemize}
	Thus, we establish consistency,~$\sqrt{n}$-convergence, and standard Gaussian inference under the weakest possible conditions on the allowed fraction of contaminated observations. That is, we characterize exactly the boundary below which one can conduct inference under contamination \emph{as if} one had access to the clean data, but above which this is not possible. Put differently, we show when and how one gets robustness for free. 
	
	Special care is needed in constructing a  positive semi-definite consistent heteroskedasticity-robust estimator of the asymptotic covariance matrix (which is built from population moments of the clean data) based on the observed contaminated data. Specifically, one cannot directly use standard methods based on the empirical error terms and sample-average based constructions are clearly ruled out in the presence of contamination. We also provide a robustifed version of the Anderson-Rubin (AR) test that replaces sample averages by suitably winsorized averages. This winsorized AR-test is shown to be valid under no/weak identification, adversarial contamination, and heteroskedasticity if~$\sqrt{n}\eta_n^{1-\frac1m}\to 0$.
	
	Finally, we study the finite sample (fixed~$n$) deviation properties of the 2SLS and W-2SLS estimators in a setting of bounded second moments and strong identification. We show that even absent contamination, the 2SLS estimator must have a uniform deviation radius that can be very large by establishing lower bounds on it. This means that any interval of half-length~$t$ centered at the 2SLS estimator that contains the true parameter with probability at least~$1-\delta$, must have $t=t_\delta$ growing as rapidly as $\sqrt{\delta^{-1}}$ in~$\delta$ as~$\delta\to 0$. This fact is driven by the poor concentration properties of sample averages already exhibited in~\cite{catoni2012challenging}. Our finite-sample study provides a theoretical foundation for the observation in~\cite{young2022consistency}  that ``2SLS may be prized for its asymptotic consistency, but in finite samples it often allows for very little inference.'' On the other hand, we show that the uniform deviation radius of the W-2SLS grows much slower, $\sqrt{\log(\delta^{-1})}$, in~$\delta$ as~$\delta\to 0$ than the lower bound on that of the 2SLS estimator. Thus, intervals centered at W-2SLS containing the true parameter with probability at least~$1-\delta$ can be much shorter than those based on classic 2SLS. 
	
	\subsection{Related literature}
Classical robust IV and GMM work develops resistant, bounded-influence, or high-resilience procedures; see \cite{krasker1985resistant}, \cite{krasker1986two}, \cite{krishnakumar1997robust}, \cite{wagenvoort2002b}, \cite{kim2007two}, and \cite{zhelonkin2012robustness}. \cite{ronchetti2001robust} quantify local asymptotic distortions of the level and power of GMM tests, \cite{vcivzek2016generalized} develops a high-breakdown generalized method of trimmed moments, and \cite{kitamura2013robustness} construct an asymptotically minimax-robust and model-efficient estimator under shrinking neighborhoods of moment-restriction models. \cite{cohen2013natural} is closest in terms of  estimator construction: it robustifies ordinary IV using multivariate location and scatter $S$-estimates and establishes high resilience, bounded influence, consistency, and asymptotic normality. \cite{jiao2019simple} studies residual-trimmed 2SLS and compares its efficiency with ordinary 2SLS.

These papers use different notions of robustness. Local-neighborhood analyses are distributional, so they do not provide guarantees uniform over contamination rules that inspect the realized clean sample before choosing the corrupted observations and their replacement values. High-breakdown analyses allow worst-case replacement when defining the breakdown point, but the cited consistency and asymptotic-normality results are established under clean sampling. None of the papers provides estimation-error or coverage guarantees in the \emph{presence of data-dependent} contamination.

\cite{klooster2024outlier} and \cite{klooster2024resistant} construct bounded-influence AR, K, and CLR procedures. Their asymptotic distributions are established absent contamination and they study local contamination through influence functions and simulations. \cite{solvsten2020robust} studies a distinct many-instrument problem, combining LIML with winsorization of structural residuals and minimizing worst-case asymptotic variance over a contaminated-normal neighborhood.

\cite{rohatgi2022robust} and \cite{forneron2023occasionally} are closest in allowing adversarial contamination and providing quantitative contaminated-sample guarantees. For linear IV, \cite{rohatgi2022robust} obtain an estimation error bound under strong identification, covariance, hypercontractivity, eighth-moment, and conditional-noise restrictions. Their stated guarantee does not achieve a $n^{-1/2}$ rate. The implementation of their algorithm requires valid supplied identification, scale, parameter-radius, and strictly positive contamination-budget bounds. They do not cover Gaussian or studentized inference, matching lower bounds, or weak-identification-robust inference. Imposing a polynomial envelope on outlying moment contributions, \cite{forneron2023occasionally} obtains finite-sample bounds and oracle-equivalent asymptotic normality for adversarially contaminated GMM, including IV, while allowing dependent observations.

Our IV focus yields a sharper characterization covering both estimation and inference. Under finite $m$-th moments, W-2SLS is a closed-form, drop-in modification of ordinary 2SLS attaining an estimation-error rate of $n^{-1/2}+\eta_n^{1-1/m}$ under adversarial contamination with unrestricted replacement values. Our matching lower bounds identify the exact contamination thresholds for uniform consistency, uniform root-$n$ consistency, and centered Gaussian inference with the same first-order law as clean-sample 2SLS. We also provide feasible heteroskedasticity-robust inference, develop a winsorized Anderson-Rubin test valid under weak identification and adversarial contamination, establish non-existence of honest confidence sets adaptive to an unknown contamination fraction, and show a finite-sample concentration advantage of W-2SLS over 2SLS even absent contamination. Thus, the paper characterizes not only how to robustify 2SLS, but the exact statistical price of sample-adaptive contamination for estimation and inference. 	

\section{Adversarial contamination and winsorized 2SLS}\label{sec:AdvContam}
 Consider a setting with a clean real outcome~$y_i$, a~$K\times 1$ vector of potentially endogenous clean regressors~$x_i$ and an~$L\times 1$ vector of clean instruments~$z_i$ for~$i=1,\hdots,n$ and~$L\geq K$.\footnote{All random variables are defined on a measurable space $(\Omega,\mathcal F)$. The probability measure on this space is denoted by $\P$, and expectations with respect to $\P$ are denoted by $\E$.} Throughout,~$K$ and~$L$ are fixed natural numbers and~$n\to\infty$ in all asymptotic statements. We are interested in estimation and inference on~$\beta\in\R^K$ in 
\begin{equation}\label{eq:multivariate}
y_i=x_i'\beta +u_i,\qquad i=1,\hdots,n.	
\end{equation}  
Writing $V_i=(y_i,x_i', z_i')'\in\R^{1+K+L}$,  an adversary inspects the clean sample~$V=(V_1,\hdots,V_n)$ and returns a contaminated sample~$\tilde{V}_i=(\tilde{y}_i,\tilde{x}_i', \tilde{z}_i')'\in\R^{1+K+L},\ i=1,\hdots,n$ to the researcher. We stress that the researcher only observes the contaminated sample~$\tilde{V}=(\tilde{V}_1,\hdots,\tilde{V}_n)$, but not clean~$V$. Thus any (feasible) statistical procedure takes the contaminated sample~$\tilde{V}$ as input. The \emph{identity} of the corrupted observations 
\begin{equation*}
\mathcal{O}=\mathcal{O}(V):=\cbr[1]{i\in\cbr[0]{1,\hdots,n}:\tilde{V}_i\neq V_i},
\end{equation*}
as well as their \emph{values}, i.e.,~the value of~$\cbr[0]{\tilde{V}_i}_{i\in \mathcal{O}}$, can (but need not) depend on the uncontaminated~$V$ and any other variables (e.g.,~external randomization) that the adversary but not necessarily the researcher observes.\footnote{The adversary is also allowed to contaminate the error terms~$u_i$ for~$i=1,\hdots,n$. However, since any statistical procedure can be a function of only the \emph{observed}~$\tilde{V}_i=(\tilde{y}_i,\tilde{x}_i, \tilde{z}_i)$, this is of no relevance.} In particular,~$\mathcal{O}$ can be a random subset of~$\cbr[0]{1,\hdots,n}$. We stress that the adversary is a device for defining a worst-case guarantee and need not represent a literal malicious actor. This \emph{adversarial contamination} imposes no restrictions on \emph{how} the observations are contaminated and is thus well suited to economic data, where contamination need not be random: influential observations may reflect selective reporting, strategic mismeasurement, coding choices, survey manipulation, or institutional features that affect precisely those units with high leverage. In contrast to the classical Huber model, which treats contamination as arising from a fixed mixture mechanism, adversarial contamination allows both the identities and the values of the contaminated observations to depend on the clean sample~$V$.

Although no assumptions are imposed on \emph{how} observations are contaminated, we restrict \emph{how many} observation can be contaminated. To be precise, we assume that at most~$\eta_n n$ of~$V_1,\hdots,V_n$ are changed, that is
\begin{equation}\label{eq:contamfrac}
|\mathcal{O}(V)|\leq \eta_n n,
\end{equation}
where~$\eta_n \in [0, 1]$ is non-random.  We stress that $\eta_n$ is a researcher-specified \emph{robustness budget}, not the unknown realized fraction~$|\mc{O}|/n$ of altered observations. Equivalently, the robustness budget may be specified as an integer $q_n$: setting~$\eta_n=q_n/n$ protects against arbitrary alterations of at most $q_n$ observations. 
  \begin{example}[Test score manipulation]
  	Suppose that $y_i$ is  the uncontaminated test score of individual~$i$, $c$ is the passing threshold, and an
institution wishes to attain a pass rate of at least $\rho\in(0,1)$. Such incentives can generate manipulation concentrated just below consequential
score thresholds; see \cite{dee2019causes} and, in an IV setting, \cite{angrist2017small}. Let $q_n=\lfloor\eta_n n\rfloor$  and (for~$x_+:=\max(x,0)$ for~$x\in\R$)
\begin{equation*}
m(V):=\min\cbr[4]{
q_n,\del[3]{\lceil\rho n\rceil-
\sum_{i=1}^n\{y_i\geq c\}}_+
}.	
\end{equation*}
Let $O(V)$ contain the $m(V)$ observations with $y_i<c$ closest to $c$,
with ties broken by the index, and define
\begin{align*}
\widetilde V_i=
\begin{cases}
(c,x_i',z_i')',& i\in O(V),\\
V_i,& i\notin O(V).	
\end{cases}
\end{align*}

Thus, whether manipulation occurs, how many scores are changed, and which scores
are selected depend on the clean sample $V$. Moreover,
$|O(V)|=m(V)\leq q_n\leq\eta_n n$, so \eqref{eq:contamfrac} is satisfied.
The example illustrates that adversarially contaminated observations need not be conventional outliers: the manipulated scores are placed exactly at, or just above, an institutionally relevant threshold~$c$.
\end{example}

It is tempting to require~$\eta_n$ in~\eqref{eq:contamfrac} large in order to be robust against many contaminated observations. However, Theorem~\ref{thm:NoAdapt} below shows that confidence sets for~$\beta$ with uniform coverage guarantees must have a diameter that depends on the largest fraction of contaminated observations~$\eta_n$ that one wishes the confidence sets to maintain coverage guarantees under. Thus, \emph{any} procedure must explicitly or implicitly take a stance on~$\eta_n$ --- the usual clean-data validity guarantees for 2SLS or the Anderson-Rubin test correspond to the benchmark~$\eta_n=0$. 

The conditions we impose on~$\eta_n$ below depend on the statistical objective: consistency of an estimator,~$\sqrt{n}$-convergence, or first-order Gaussian inference with the same law as under clean sampling. Section~\ref	{sec:Opt} proves that for each of these, the assumptions that we impose on~$\eta_n$ are the weakest possible by providing matching lower bounds. Finally, we stress that, under the conditions of Theorem \ref{thm:limdist}, protection against any fixed number~$q_n$ of contaminated observations is asymptotically \emph{free} at first order, which is already relevant in light of~\cite{young2022consistency}.
 
\subsection{The winsorized 2SLS estimator} 
The classic 2SLS estimator of~$\beta$ in~\eqref{eq:multivariate} is (cf., e.g.,~Chapter 5 in \cite{wooldridge2010econometric}) 
\begin{align}\label{eq:2SLSM}
	\hat{\beta}_{n,\text{2SLS}}
	=
	\hat{\beta}_{n,\text{2SLS}}(V)
	=
	\del[3]{\frac{X'Z}{n}\sbr[2]{\frac{Z'Z}{n}}^{-1}\frac{Z'X}{n}}^{-1}\frac{X'Z}{n}\sbr[2]{\frac{Z'Z}{n}}^{-1}\frac{Z'y}{n},
\end{align}
where~$y=(y_1,\hdots,y_n)'$,~$X=(x_1,\hdots, x_n)'$, and~$Z=(z_1,\hdots,z_n)'$ (grant all expressions are well-defined). Observe that $\hat{\beta}_{n,\text{2SLS}}$ has been defined as a function of the \emph{uncontaminated} data~$V$. Since it is based on plain sums, it is highly non-robust as these sums can be manipulated to take any value by changing just \emph{one} of the summands. We now introduce the winsorized 2SLS estimator and provide conditions under which it is close to the infeasible~$\hat{\beta}_{n,\text{2SLS}}$.

Our winsorized 2SLS estimator simply replaces the non-robust sample averages entering the definition of the 2SLS estimator in~\eqref{eq:2SLSM} by quantile-winsorized means. To make this precise, for real numbers~$s_1,\hdots,s_n$, denote by~$s_1^*\leq \hdots\leq s_n^*$ their non-decreasing rearrangement. Let~$-\infty<\alpha\leq\beta<\infty$ and define~$\phi_{\alpha,\beta}:\R\to [\alpha,\beta]$ as
\begin{equation}\label{eq:winsor}
\phi_{\alpha,\beta}(s)
:=
\begin{cases}
\alpha\qquad \text{if }s<\alpha\\
s\qquad \text{if }s\in[\alpha,\beta]\\
\beta\qquad \text{if }s>\beta.
\end{cases}
\end{equation} 
Now, for~$\eps_n\in(0,1/2]$, define the quantile-winsorized mean~$\mu_{\eps_n}:\R^n\to [\alpha,\beta]$ as
\begin{equation}\label{eq:winsfunc}
\mu_{\eps_n}(s_1,\hdots,s_n)
=
\frac{1}{n}\sum_{i=1}^{n}
\phi_{s_{(\lceil \eps_n n\rceil)}^*,s_{(\lfloor(1-\eps_n)n\rfloor)+1}^*}(s_i).	
\end{equation}
Thus,~$\mu_{\eps_n}(s_1,\hdots,s_n)$ is the ordinary arithmetic mean after replacing observations below the empirical $\eps_n$-quantile by that quantile and observations above the empirical $(1-\eps_n)$-quantile by that quantile. Given the contaminated sample~$\tilde{V}$, let~$W_{ZX}$,~$W_{ZZ}$, and~$W_{Zy}$ be the~$L\times K$,~$L\times L$ and~$L\times 1$ matrices with entries
\begin{align*}
	W_{ZX,l,k}&=\mu_{\eps_n}\del[1]{\tilde{z}_{1,l}\tilde{x}_{1,k},\hdots,\tilde{z}_{n,l}\tilde{x}_{n,k}},\qquad
	W_{ZZ,l,j}=\mu_{\eps_n}\del[1]{\tilde{z}_{1,l}\tilde{z}_{1,j},\hdots,\tilde{z}_{n,l}\tilde{z}_{n,j}},\\
	W_{Zy,l}&=\mu_{\eps_n}\del[1]{\tilde{z}_{1,l}\tilde{y}_{1},\hdots,\tilde{z}_{n,l}\tilde{y}_{n}},\qquad\text{for}\qquad l,j=1,\hdots,L\text{ and }k=1,\hdots,K
\end{align*}
 The winsorized 2SLS estimator now replaces the non-robust~$Z'X/n$,~$Z'Z/n$ and~$Z'y/n$ entering construction of~$\hat{\beta}_{n,\text{2SLS}}$ by~$W_{ZX}$,~$W_{ZZ}$, and~$W_{Zy}$, respectively, that is:\footnote{The proof of Proposition~\ref{prop:IVClose} establishes that~$W_{ZZ}$ and~$W_{ZX}'[W_{ZZ}]^{-1}W_{ZX}$ are invertible for~$n$ sufficiently large. We set~$\hat{\beta}_{n,\text{W-2SLS}}=0_K$ when either of $W_{ZZ}$ and~$W_{ZX}'[W_{ZZ}]^{-1}W_{ZX}$ is not invertible.}
 \begin{equation}\label{eq:W2SLSM}
	\hat{\beta}_{n,\text{W-2SLS}}
	=
	\hat{\beta}_{n,\text{W-2SLS}}(\tilde{V})
	=
	\del[1]{W_{ZX}'[W_{ZZ}]^{-1}W_{ZX}}^{-1}W_{ZX}'[W_{ZZ}]^{-1}W_{Zy}.
\end{equation}
$\hat{\beta}_{n,\text{W-2SLS}}$ has exactly the same structure as~$\hat{\beta}_{n,\text{2SLS}}$, but uses the robust winsorized~$\mu_{\eps_n}$ instead of sample averages and takes the contaminated~$\tilde{V}$ as input instead of the clean~$V$.
  
When allowing up to~$\eta_n n$ observations to be contaminated, it is clear that even~$\hat{\beta}_{n,\text{W-2SLS}}$ can perform poorly unless at least the smallest and largest~$\eta_n n$ observations are winsorized in~\eqref{eq:winsfunc}. Thus, one must choose~$\eps_n> \eta_n$, implying in particular that~$\eta_n\leq1/2$ must hold. It is also clear that winsorizing substantially more than~$\eta_nn$ observations is wasteful and we shall establish that~$\hat{\beta}_{n,\text{W-2SLS}}$ with
\begin{equation}\label{eq:epsfam}
\eps_n=
\eps_n(\eta_n)
:=
1.01\cdot \eta_n +\lambda\cdot \frac{\log(n)}{n}
,\qquad \text{for fixed }\lambda\in (0,\infty)
\end{equation} 
has several optimality properties. Since~$\eta_n\to 0$ in all asymptotic statements, it holds that~$\eps_n\to 0$ implying that~$\eps_n\in(0,1/2]$ eventually. 
\begin{assumption}\label{ass:sequence}
$L\geq K$ and the clean data satisfies:
	\begin{enumerate}
		\item $\cbr[0]{x_i,z_i,u_i}_{i\in\N}$ is i.i.d.~satisfying~\eqref{eq:multivariate};
		\item $\E (|z_{1,l}z_{1,j}|^m)<\infty$,~$\E(|z_{1,l}x_{1,k}|^m)<\infty$, and $\E (|u_1z_{1,l}|^m)<\infty$ for~$l,j=1,\hdots,L$,~$k=1,\hdots,K$, and some~$m>1$;		
		\item $\text{rank}[\E(z_1z_1')]=L$ and~$\text{rank}[\E(z_1x_1')]=K$.
	\end{enumerate}	
\end{assumption}
 Assumption~\ref{ass:sequence} imposes classic textbook assumption used to establish positive results on the clean-data 2SLS estimator, cf., e.g.,~\cite{wooldridge2010econometric}. In Theorem~\ref{thm:estim} and onwards, we also impose~$\E(z_{1}u_1)=0$, but Proposition \ref{prop:IVClose} does not need it.

\subsection{When is robustness for free?}\label{sec:WinsMain}
We first show that the winsorized 2SLS estimator~$\hat{\beta}_{n,\text{W-2SLS}}(\tilde{V})$ based on the \emph{contaminated} data is close to the infeasible ordinary 2SLS estimator~$\hat{\beta}_{n,\text{2SLS}}(V)$ based on the \emph{clean} data. Let~$||\cdot||$ denote the Euclidean norm on~$\R^K$.
\begin{proposition}\label{prop:IVClose}
	Suppose that Assumption~\ref{ass:sequence} is satisfied and let~$\eps_n$ be chosen as in~\eqref{eq:epsfam}. Then, if~\eqref{eq:contamfrac} is satisfied with~$\eta_n\to 0$ as~$n\to\infty$, it holds that
	\begin{equation*}
		\enVert[1]{\hat{\beta}_{n,\text{W-2SLS}}-\hat{\beta}_{n,\text{2SLS}}}
		=
		O_\P\del[3]{\eta_n^{1-\frac{1}{m}}+\sbr[2]{\frac{\log(n)}{n}}^{1-\frac{1}{m}}}
		=
		o_\P(1).
	\end{equation*}
		\end{proposition}
Proposition~\ref{prop:IVClose} shows that the difference between~$\hat{\beta}_{n,\text{W-2SLS}}$ and~$\hat{\beta}_{n,\text{2SLS}}$ converges to zero in probability when~$\eta_n\to 0$. Because
\begin{equation*}
\enVert[1]{\hat{\beta}_{n,\text{W-2SLS}}-\beta}
\leq
\enVert[1]{\hat{\beta}_{n,\text{W-2SLS}}-\hat{\beta}_{n,\text{2SLS}}}+\enVert[1]{\hat{\beta}_{n,\text{2SLS}}-\beta},	\end{equation*} 
and
\begin{equation*}
\enVert[1]{\hat{\beta}_{n,\text{2SLS}}-\beta}
\leq
\enVert[1]{\hat{\beta}_{n,\text{2SLS}}-\hat{\beta}_{n,\text{W-2SLS}}}+\enVert[1]{\hat{\beta}_{n,\text{W-2SLS}}-\beta},	\end{equation*}
it follows that~$\hat{\beta}_{n,\text{W-2SLS}}(\tilde{V})$ is consistent in the presence of a moderate amount of contamination ($\eta_n\to 0$) \emph{if and only if}~$\hat{\beta}_{n,\text{2SLS}}(V)$ is consistent on the clean sample. Thus, in the sense of consistency, one obtains robustness of the winsorized 2SLS estimator for free if~$\eta_n\to 0$. 

The following theorem also imposes~$\E(z_{1}u_1)=0$ and shows that not only the probability limit, but also the rate of convergence of~$\hat{\beta}_{n,\text{W-2SLS}}$ to~$\beta$ is the same $\sqrt{n}$-rate as that of the infeasible~$\hat{\beta}_{n,\text{2SLS}}$ if~$\eta_n^{1-\frac1m}\lesssim n^{-1/2}$.\footnote{For non-negative real sequences~$(a_n)_{n\in\N}$ and~$(b_n)_{n\in\N}$,~$a_n\lesssim b_n$ means that there exists a~$C>0$ such that~$a_n\leq Cb_n$ for~$n$ suffiently large.} 
\begin{theorem}\label{thm:estim}
Suppose that Assumption~\ref{ass:sequence} is satisfied for some~$m>2$ and let~$\eps_n$ be chosen as in~\eqref{eq:epsfam}. Suppose that~$\E(z_{1}u_1)=0$. Then, if~\eqref{eq:contamfrac} is satisfied with~$\eta_n\to 0$ as~$n\to\infty$, it holds that
\begin{equation*}
\enVert[1]{\hat{\beta}_{n,\text{W-2SLS}}-\beta}
=
O_\P\del[2]{\eta_n^{1-\frac{1}{m}}+n^{-1/2}}
=
o_\P(1).	
\end{equation*}	
\end{theorem} 
Theorem~\ref{thm:estim} shows that absent contamination, i.e.~when~$\eta_n=0$, $\hat{\beta}_{n,\text{W-2SLS}}$ has the same~$\sqrt{n}$-convergence rate as the infeasible~$\hat{\beta}_{n,\text{2SLS}}$. More importantly, this convergence rate is achieved as long as~$\eta_n^{1-\frac1m} \lesssim n^{-1/2}$. Thus, below this threshold one obtains robustness for free in the rate-of-convergence sense. If~$\eta_n^{1-\frac1m} \gtrsim n^{-1/2}$, then~$\hat{\beta}_{n,\text{W-2SLS}}$ is still consistent at rate~$\eta_n^{1-\frac1m}$ if~$\eta_n\to 0$. Theorem~\ref{thm:contamLB} below provides the corresponding uniform lower bound showing that \emph{no} estimator can have a uniform convergence rate with a better dependence on~$\eta_n$ than~$\eta_n^{1-\frac1m}$. Since the usual parametric~$n^{-1/2}$ rate is already unimprovable in the uncontaminated scalar location model, the rate~$\eta_n^{1-\frac{1}{m}}+n^{-1/2}$ achieved by~$\hat{\beta}_{n,\text{W-2SLS}}$ is minimax sharp. In particular, no estimator can be uniformly consistent if~$\limsup_{n\to\infty}\eta_n>0$ and the condition~$\eta_n\to 0$ in Theorem~\ref{thm:estim} is therefore necessary for the existence of a uniformly consistent estimator.        

Proposition~\ref{prop:IVClose} with~$m>2$ and~$\sqrt{n}\eta_n^{1-\frac1m}\to 0$ also implies the asymptotic equivalence
\begin{equation*}
\enVert[1]{\sqrt{n}(\hat{\beta}_{n,\text{W-2SLS}}-\hat{\beta}_{n,\text{2SLS}})}=o_\P(1).	
\end{equation*}
Hence, $\sqrt{n}\del[1]{\hat{\beta}_{n,\text{W-2SLS}}-\beta}$ and~$\sqrt{n}\del[1]{\hat{\beta}_{n,\text{2SLS}}-\beta}$ have the same limiting distribution whenever the 2SLS estimator based on the clean data has one. 

\begin{theorem}\label{thm:limdist}
Suppose that Assumption~\ref{ass:sequence} is satisfied for some~$m>2$ and let~$\eps_n$ be chosen as in~\eqref{eq:epsfam}. Suppose that $\E(z_{1}u_1)=0$ and~$\text{rank}\sbr[1]{\E(u_1^2z_1z_1')}=L$. Then, if~\eqref{eq:contamfrac} is satisfied with~$\sqrt{n}\eta_n^{1-\frac1m}\to 0$ as~$n\to\infty$, 
$\sqrt{n}\del[1]{\hat{\beta}_{n,\text{W-2SLS}}-\beta}$ and~$\sqrt{n}\del[1]{\hat{\beta}_{n,\text{2SLS}}-\beta}$ both converge in distribution to
		\begin{equation*}
			\mathsf{N}_K\del[1]{0,\Omega},\quad\text{where}\quad \Omega=A\E(u_1^2z_1z_1')A',
		\end{equation*}
		for~$A=\del[1]{\E(z_1x_1')'\sbr[0]{\E(z_1z_1')}^{-1}\E(z_1x_1')}^{-1}\E(z_1x_1')'\sbr[0]{\E(z_1z_1')}^{-1}$.
	\end{theorem}
Theorem~\ref{thm:limdist} shows when robustness is for free in the first-order asymptotic sense: If~$\sqrt{n}\eta_n^{1-\frac1m}\to 0$ then $\hat{\beta}_{n,\text{W-2SLS}}$ has exactly the same asymptotic distribution as $\hat{\beta}_{n,\text{2SLS}}$ based on the \emph{clean} sample. Hence the price of robustness is zero in this regime: the confidence sets have the same asymptotic size as clean 2SLS confidence sets, while W-2SLS remains consistent under an arbitrary but vanishing fraction of observations, allowing the number of such observations to grow with the sample size. 

By Theorem~\ref{thm:contamLB} below, the condition~$\sqrt{n}\eta_n^{1-\frac1m}\to 0$ in Theorem~\ref{thm:limdist} is necessary for the $\sqrt{n}$-scaled error of \emph{any} estimator to converge uniformly in distribution to a mean-zero non-degenerate normal distribution. That is,~$\sqrt{n}\eta_n^{1-\frac1m}\to 0$ is the weakest possible condition on~$\eta_n$ for conducting inference as if no contamination were present.

Recall that \cite{young2022consistency} showed that deleting $q_n=2$ observations can overturn a large share of nominally (in)significant 2SLS results. This is exactly the regime where $\eta_n=q_n/n$ is small, yet empirically important. If~$\eta_n=q/n$ for some fixed~$q\in\cbr[0]{0}\cup \N$, then~$\sqrt{n}\eta_n^{1-\frac1m}=q^{1-\frac1m}\cdot n^{-\frac12+\frac1m}\to 0$, implying that the conditions of Theorem~\ref{thm:limdist} are satisfied for all~$m>2$.

\subsection{Estimating~$\Omega$, hypothesis testing, and confidence intervals}
To conduct feasible inference based on the Gaussian approximation to the distribution of~$\sqrt{n}\del[1]{\hat{\beta}_{n,\text{W-2SLS}}-\beta}$ in Theorem~\ref{thm:limdist}, a consistent estimator of the covariance matrix~$\Omega$ is needed. We now construct a consistent positive semi-definite estimator of~$\Omega$ taking the contaminated data as input.

 In the presence of contamination, it is clear that estimators based on plain sample averages are ruled out. Note also that already for the purpose of estimating the variance~$\sigma^2=\E(u_1^2)$ of the~$u_i$ \emph{absent} contamination, the usual consistency argument for~$n^{-1}\sum_{i=1}^n\hat{u}_i^2$ where~$\hat{u}_i=y_i-x_i'\hat{\beta}_{n,\text{2SLS}}$ relies heavily on the linearity of the arithmetic mean: Specifically, using that~$\hat{u}_i=u_i-x_i'(\hat{\beta}_{n,\text{2SLS}}-\beta)$ such that 
\begin{equation*}
\frac{1}{n}\sum_{i=1}^n\hat{u}_i^2
=
\frac{1}{n}\sum_{i=1}^n u_i^2+\frac{1}{n}\sum_{i=1}^n[x_i'(\hat{\beta}_{n,\text{2SLS}}-\beta)]^2-\frac{2}{n}\sum_{i=1}^nx_i'u_i(\hat{\beta}_{n,\text{2SLS}}-\beta),	
\end{equation*}
implying that consistency of~$n^{-1}\sum_{i=1}^n\hat{u}_i^2$ follows from the first average on the right-hand side converging to~$\sigma^2$ and the two last ones to zero. However, writing~$\hat{\tilde{u}}_i=\tilde{y}_{i}-\tilde{x}_i'\hat{\beta}_{n,\text{W-2SLS}}$ for the contaminated residuals based on the robust estimator~$\hat{\beta}_{n,\text{W-2SLS}}$ and letting~$\tilde{u}_i:=\tilde{y}_{i}-\tilde{x}_i'\beta$, the non-linearity of the winsorization function~$\phi_{\alpha,\beta}$ defined in~\eqref{eq:winsor} implies that the previous decomposition does not apply to~$\mu_{\eps_n}(\hat{\tilde{u}}_1^2,\hdots,\hat{\tilde{u}}_n^2)$, making this candidate estimator difficult to analyze.
 The same difficulty remains when estimating the component~$\Omega_S:=\E(u_1^2z_1z_1')$ of the covariance matrix~$\Omega=A\Omega_SA'$ instead of ``just''~$\sigma^2$. 

Our solution is to avoid winsorizing the generated residuals~$\hat{\tilde{u}}_i$. To be precise, for~$b$ being a~$K\times 1$ vector and~$\otimes$ denoting the Kronecker product, let
\begin{equation*}
	\mathfrak{p}(b)=(1,-b'),\quad \mathfrak{P}(b)=I_L\otimes \mathfrak{p}(b),\quad\tilde{w}_i=(\tilde{y}_i,\tilde{x}_i')',\quad\text{and}\quad \tilde{h}_i=\tilde{z}_i\otimes \tilde{w}_i,\quad i=1,\hdots,n.
\end{equation*}   
Similarly,~$w_i=(y_i,x_i')'$ and~$h_i=z_i\otimes w_i$ such that 
(by the mixed product property of Kronecker products)
\begin{equation*}
	\mathfrak{P}(\beta)\tilde{h}_i=\tilde{z}_i\otimes \mathfrak{p}(\beta)\tilde{w}_i=\tilde{z}_i\otimes\tilde{u}_i=\tilde{z}_i\tilde{u}_i\qquad\text{and}\qquad 	\mathfrak{P}(\beta)h_i=z_i\otimes \mathfrak{p}(\beta)w_i=z_i\otimes u_i= z_iu_i,
\end{equation*}
implying that
\begin{equation*}
	\Omega_S=\E(u_1^2z_1z_1')=\E\del[1]{\mathfrak{P}(\beta)h_1h_1'\mathfrak{P}(\beta)'}=\mathfrak{P}(\beta)\E(h_1h_1')\mathfrak{P}(\beta)'.
\end{equation*}
Thus, since~$\hat{\beta}_{n,\text{W-2SLS}}$ converges to~$\beta$ in probability under the conditions of Theorem~\ref{thm:estim} and~$b\mapsto \mathfrak{P}(b)$ is continuous, it suffices to exhibit a consistent and positive semidefinite estimator of~$\Xi=\E(h_1h_1')$ based on the \emph{contaminated}~$\tilde{h}_i$ in order to estimate~$\Omega_S$ consistently (and in a positive semi-definite fashion). It would be natural to estimate~$\Xi$ entry-by-entry via~$\mu_{\eps_n}(\tilde h_{1,j}\tilde h_{1,k},\hdots,\tilde h_{n,j}\tilde h_{n,k})$ for each pair~$1\leq j,k\leq L(K+1)$, cf.~\eqref{eq:winsfunc}. However, although consistent for~$\Xi$ by Theorem 2.1 in~\cite{Wins1}, this estimator need not be positive semi-definite at every sample size~$n$. To solve this issue, we propose the following positive semidefinite estimator of~$\Xi$:

Given~$s=(s_1,\hdots,s_n)\in\R^n$ and~$\eps_n\in(0,1/2]$, let~$\phi_i:\R^n\to \R$ be defined as $\phi_i(s)=\phi_{s_{(\lceil \eps_n n\rceil)}^*,s_{(\lfloor(1-\eps_n)n\rfloor)+1}^*}(s_i)$ and let~$\hat{\Xi}$ be the~$L(K+1)\times L(K+1)$ matrix with entries
\begin{equation}\label{eq:Xiestim}
	\hat{\Xi}_{j,k}=\frac{1}{n}\sum_{i=1}^n \phi_i(\tilde{h}_{1,j},\hdots,\tilde{h}_{n,j})\cdot \phi_i(\tilde{h}_{1,k},\hdots,\tilde{h}_{n,k}),\qquad 1\leq j,k\leq L(K+1) 
\end{equation}
Note that~$\hat{\Xi}$ is symmetric and positive semi-definite by virtue of being a Gram matrix. Lemma~\ref{lem:productwins} in the appendix yields the consistency of~$\hat{\Xi}$ by showing that each of its entries is close to the consistent entry-by-entry estimator $\mu_{\eps_n}(\tilde h_{1,j}\tilde h_{1,k},\hdots,\tilde h_{n,j}\tilde h_{n,k})$. Now define the plug-in estimators of~$\Omega_S$ and~$A$ as
\begin{equation*}
	\hat{\Omega}_S
	=
	\mathfrak{P}(\hat{\beta}_{n,\text{W-2SLS}})\hat{\Xi}\mathfrak{P}(\hat{\beta}_{n,\text{W-2SLS}})' \qquad\text{and}\qquad \hat{A}=\del[1]{W_{ZX}'[W_{ZZ}]^{-1}W_{ZX}}^{-1}W_{ZX}'[W_{ZZ}]^{-1},\end{equation*}
and finally define the positive semi-definite estimator of~$\Omega$ 
\begin{equation*}
\hat{\Omega}
=
\hat{A}\hat{\Omega}_S\hat{A}'.\footnote{On the event that the displayed inverses do not exist, one may define~$\hat{A}$ and~$\hat{\Omega}$ arbitrarily. By the proof of Theorem~\ref{thm:FeasibleInference} this event has probability converging to zero.}	
\end{equation*}
\begin{theorem}\label{thm:FeasibleInference}
Suppose that Assumption~\ref{ass:sequence} is satisfied for some~$m>2$ and let~$\eps_n$ be chosen as in~\eqref{eq:epsfam}. Suppose that $\E(z_{1}u_1)=0$ and~$\text{rank}\sbr[1]{\E(u_1^2z_1z_1')}=L$. Then, if~\eqref{eq:contamfrac} is satisfied with~$\eta_n\to 0$, it holds that~$\hat{\Omega}$ converges to~$\Omega$ in probability. Therefore,~$\hat{\Omega}^{-1/2}$ is well-defined with probability tending to one. 

If, in addition~$\sqrt{n}\eta_n^{1-\frac1m}\to 0$, then
\begin{equation}\label{eq:feasible}
	\hat{\Omega}^{-1/2}\sqrt{n}\del[1]{\hat{\beta}_{n,\text{W-2SLS}}-\beta}\qquad\text{converges in distribution to}\qquad \mathsf{N}_K(0,I_K).
\end{equation}
\end{theorem}
Theorem~\ref{thm:FeasibleInference} justifies conducting standard asymptotic inference based on a Gaussian approximation even in the presence of contamination when~$\sqrt{n}\eta_n^{1-\frac1m}\to 0$. The limiting distribution of~$\hat{\Omega}^{-1/2}\sqrt{n}\del[1]{\hat{\beta}_{n,\text{W-2SLS}}-\beta}$ is the same~$\mathsf{N}_K(0,I_K)$ as that of the infeasible 2SLS estimator premultiplied by the standard estimator of~$\Omega^{-1/2}$.

%
\begin{remark}
	Fix~$k\in\cbr[0]{1,\hdots,K}$ and consider testing~$\mathsf{H}_0:\beta_k=\beta_k^0$ against~$\mathsf{H}_1:\beta_k\neq \beta_k^0$.~\eqref{eq:feasible} of Theorem~\ref{thm:FeasibleInference} implies that under~$\mathsf{H_0}$,
	\begin{equation*}
	t_k:=\frac{\sqrt{n}(\hat{\beta}_{n,\text{W-2SLS},k}-\beta_k^0)}{\sqrt{\hat{\Omega}_{k,k}}}	\qquad \text{converges in distribution to}\qquad \mathsf{N}_1(0,1),
	\end{equation*}
	so that~$t$-tests can be conducted using Gaussian critical values and the robust~$t$-statistic~$t_k$. Letting~$z_{1-\delta}:=\Phi^{-1}(1-\delta)$ for~$\delta\in(0,1)$ and~$\Phi$ be the cdf of the standard normal the (sequence of) test(s)~$\mathds{1}\cbr[0]{|t_k|>z_{1-\delta/2}}$ has asymptotic size~$\delta$ as the standard tests based on the infeasible 2SLS estimator using the \emph{uncontaminated} data. Thus, robustness is free in the first-order asymptotic sense. Similarly, the (sequence of) confidence set(s)
	\begin{equation*}
		\hat{C}_{k,1-\delta}:=\sbr[4]{\hat{\beta}_{n,\text{W-2SLS},k}-z_{1-\delta/2}\sqrt{\frac{\hat{\Omega}_{k,k}}{n}},\hat{\beta}_{n,\text{W-2SLS},k}+z_{1-\delta/2}\sqrt{\frac{\hat{\Omega}_{k,k}}{n}}},\qquad k=1,\hdots,K,
	\end{equation*}
	satisfies that~$\P\del[1]{\beta_k\in \hat{C}_{k,1-\delta}}\to 1-\delta$ as~$n\to\infty$ for every~$k=1,\hdots,K$ without being asymptotically longer than the usual confidence intervals at the~$\sqrt{n}$-scale, 
	
	Joint inference on a subset of entries of~$\beta$ can also be conducted in a standard fashion based on Theorem~\ref{thm:FeasibleInference} in the presence of adversarial contamination without asymptotic efficiency loss compared to inference based on the standard 2SLS estimator implemented on the \emph{uncontaminated} data.\qed
\end{remark}

\section{Optimality of the dependence on~$\eta_n$}\label{sec:Opt}
Section \ref{sec:WinsMain} imposed different conditions on~$\eta_n$ for different statistical objectives. We now show that these conditions are the weakest possible in a minimax sense. The condition~$\eta_n\to 0$ is necessary for uniform consistency;  $\sqrt{n}\eta_n^{1-1/m}=O(1)$ is the threshold for uniform~$\sqrt{n}$-consistency; and~$\sqrt{n}\eta_n^{1-1/m}\to 0$ is necessary for centered Gaussian inference with the same first-order law as in the uncontaminated model. We then show that honest confidence sets cannot adapt to the unknown contamination fraction; even absent contamination the diameter must depend on the largest fraction~$\eta_n$ of contaminated observations that one is willing to entertain.
 
Specifically, consider the just identified case of~$L=K=1$. For any distribution~$Q$ of the uncontaminated~$R_i=(y_i,x_i,z_i,u_i)$ on~$\R^4$ and~$R=(R_1,\hdots,R_n)$, let~$P_Q$ be a measure on the underlying measurable space~$(\Omega,\mc{F})$ such that~$P_Q\circ R^{-1}=Q^{\otimes n}$.\footnote{\label{fn:prodspace}This requires that the underlying probability space~$(\Omega,\mc{F})$ is rich enough to support such~$P_Q$, but is without loss of generality for the purpose of the lower bounds in Theorems~\ref{thm:contamLB} and~\ref{thm:NoAdapt} where, as usual, one can let~$(\Omega,\mc{F},P_Q)=((\R^{4})^n,\mc{B}(\R^{4})^{\otimes n},Q^{\otimes n})$ and let the~$R_i$ be the~$i$th canonical coordinate projections (onto~$\R^4$). This also applies to the constructions in Section~\ref{sec:FiniteDev}.} That is, the~$R_i$ are i.i.d.~with common distribution~$Q$.
Let~$m\geq 2$ and denote by~$\mc{Q}_m$ the family of distribution on~$\R^{4}$, each element~$Q$ of which satisfies for some~$\beta(Q)\in[0,1]$ that
\begin{align}\label{eq:LBFam}
&P_Q(y_1=\beta(Q)x_1+u_1, x_1=1,z_1=1)=1,\quad E_{P_Q}|z_1u_1|^m=E_{P_Q}|u_1|^m\leq 1.5\cdot2^m,\notag\\
&E_{P_Q}(z_1u_1)=E_{P_Q}(u_1)=0,\quad \text{and}\quad E_{P_Q}(z_1^2u_1^2)=E_{P_Q}(u_1^2)\in[1,4]. 	
\end{align}
Letting~$x_i=z_i=1$ with probability one is convenient since it reduces the model~\eqref{eq:multivariate} to a location (sub)model for the~$y_i$, yet  satisfies the relevant conditions of Theorem~\ref{thm:estim}--\ref{thm:limdist}. 

Finally, it suffices to consider a specific contamination scheme changing at most~$\eta_n n$ of the observations. Specifically, we utilize an adversary who sets as many $y_i$ equaling~$1+p_n^{-1/m}$ to one as the contamination budget allows, where~$p_n=\eta_n/2$. The adversary leaves all other variables untouched, and does not use external randomization. To be precise, write~$N_i(R)
 ={\sum_{j=1}^i\mathds{1}_{\cbr[0]{y_j=1+p_n^{-1/m}}}}$. Define the measurable (contamination) mapping~$\mathsf{C}:(\R^4)^n\to(\R^3)^n$ by~$\tilde{V}:=(\tilde{V}_1,\hdots,\tilde{V}_n)=\mathsf{C}(R)$, where
\begin{equation}\label{eq:Contam}
\tilde{V}_i=(\tilde{y}_i,\tilde{x_i},\tilde{z_i})=\del[1]{y_i-p_n^{-1/m}\mathds{1}_{\cbr[0]{y_i=1+p_n^{-1/m}}}\mathds{1}_{\cbr[0]{N_i(R)\leq \eta_n n}},x_i,z_i},\qquad i=1,\hdots,n,
\end{equation} 
and~$\mathsf{C}$ drops the~$u_i$ since no statistical procedure can depend on these unobserved quantities. Because~$\mathds{1}_{\cbr[0]{N_i(R)\leq \eta_n n}}=0$ if~$N_i(R)>\eta_nn$, we note that~\eqref{eq:contamfrac} is satisfied by construction. The specific form of the contamination~$\mathsf{C}$ allows for the construction of~$Q_0,Q_1\in\mc{Q}_m$ for which~$\beta(Q_1)-\beta(Q_0)$ is of order~$\eta_n^{1-\frac1m}$ in Lemma~\ref{lem:distconstruct} but so that the total variation distance between~$Q_0^{\otimes n} \circ \mathsf{C}^{-1}$ and~$Q_1^{\otimes n} \circ \mathsf{C}^{-1}$ is suitably small (Lemma~\ref{lem:TVDist}).

\begin{theorem}\label{thm:contamLB}
Let~$L=K=1$. Fix~$n\in\N$,~$\eta_n\in(0,1]$,~$m>2$, and let
the contaminated sample~$\tilde{V}=(\tilde{V}_1,\hdots,\tilde{V}_n)$ be given by~\eqref{eq:Contam}. Then, for any measurable estimator~$T_n:\R^{3n}\to \R$ of~$\beta$ taking the contaminated sample as input, it holds that
\begin{align}\label{eq:LB}
	\sup_{Q\in\mc{Q}_m}P_Q\del[2]{\envert[1]{T_n(\tilde{V})-\beta(Q)}\geq 0.25\cdot\eta_n^{1-\frac{1}{m}}}
	\geq 
	0.5\cdot\sbr[2]{1-\del[1]{e/4}^{\eta_nn/2}}.
\end{align}
%
Furthermore, if~$\limsup_{n\to\infty}\sqrt{n}\eta_n^{1-\frac1m}>0$ then it holds for any functional~$\sigma:\mc{Q}_m\to(0,\infty)$ and~$\Phi$ being the cdf of the standard normal distribution that
\begin{equation}\label{eq:NoUnifConv}
	\limsup_{n\to\infty}\sup_{Q\in\mc{Q}_m}\sup_{t\in\R}\envert[2]{P_Q\del[1]{\sqrt{n}[T_n(\tilde{V})-\beta(Q)]\leq t}-\Phi(t/\sigma(Q))}>0
\end{equation}

\end{theorem}
Equation\eqref{eq:LB} of Theorem~\ref{thm:contamLB} shows that no estimator can have a better dependence on the contamination fraction~$\eta_n$ uniformly over~$\mc{Q}_m$ than the~$\eta_n^{1-\frac{1}{m}}$ achieved by~$\hat{\beta}_{n,\text{W-2SLS}}$ in Theorem~\ref{thm:estim}. Thus, a better dependence on~$\eta_n$ \emph{must} come at the price of imposing restrictions on the allowed contamination strategies that the adversary can employ. In particular, one must rule out strategies that are a function of the clean data~$V_1,\hdots,V_n$ (e.g.~selective reporting or survey manipulation) as exploited in the proof of Theorem~\ref{thm:contamLB}. 

Recall that Theorem~\ref{thm:estim} implies that~$\hat{\beta}_{n,\text{W-2SLS}}$ is~$\sqrt{n}$-consistent if~$\limsup_{n\to\infty}\sqrt{n}\eta_n^{1-\frac1m}<\infty$, whereas \eqref{eq:LB} of Theorem~\ref{thm:contamLB} shows that \emph{no} estimator can be uniformly $\sqrt{n}$-consistent over~$\mc{Q}_m$ in case $\limsup_{n\to\infty}\sqrt{n}\eta_n^{1-\frac1m}=\infty$. Thus,~$\hat{\beta}_{n,\text{W-2SLS}}$ is $\sqrt{n}$-consistent under the weakest possible condition on~$\eta_n$.

Finally,~\eqref{eq:NoUnifConv} of Theorem \ref{thm:contamLB} implies that the condition~$\sqrt{n}\eta_n^{1-\frac1m}\to 0$ in Theorem~\ref{thm:limdist} is the weakest possible one on~$\eta_n$ if one wishes to conduct standard inference based on a normal approximation as if no contamination were present. Specifically, if $\limsup_{n\to\infty}\sqrt{n}\eta_n^{1-\frac1m}>0$, then for any estimator~$T_n$ there exists a subsequence along which~$\sqrt{n}[T_n(\tilde{V})-\beta(Q)]$ does not converge in distribution to a normal with a mean of zero under some~$Q_n\in\mc{Q}_m$.

\subsection{Non-existence of honest adaptive confidence sets}
The following theorem establishes that any confidence set that covers  uniformly over~$\mc{Q}_m$ must have a diameter --- \emph{even absent contamination} --- that depends on the largest possible contamination fraction~$\eta_n$ that one is willing to entertain. In this sense, any procedure must explicitly or implicitly specify the largest fraction of contaminated observations that it wishes to be robust against. The choice of~$\eps_n=\eps_n(\eta_n)$ in \eqref{eq:epsfam} makes this choice explicit for the purpose of implementing the winsorized 2SLS estimator.
\begin{theorem}\label{thm:NoAdapt}
Let~$L=K=1$. Fix~$n\in\N$,~$\eta_n\in(0,1]$,~$m>2$, and~$\delta\in(0,0.5)$. Let~$C_n:(\R^3)^n\to 2^\R$ be a non-empty confidence set such that~$\cbr[0]{v\in(\R^3)^n:\beta\in C_n(v)}$ is measurable for all~$\beta\in\R$,~$V=(V_1,\hdots, V_n)$ be the uncontaminated data, and the contaminated data~$\tilde{V}=(\tilde{V}_1,\hdots,\tilde{V}_n)$ be given by~\eqref{eq:Contam}. If
\begin{equation*}
	\inf_{Q\in\mc{Q}_m}P_Q\del[1]{\beta(Q)\in C_n(V)}\geq 1-\delta\qquad \text{and}\qquad \inf_{Q\in\mc{Q}_m}P_Q\del[1]{\beta(Q)\in C_n(\tilde{V})}\geq 1-\delta,
\end{equation*}	 
then the diameter~$|C_n|:=\sup\cbr[0]{|x-y|:x,y\in C_n}$ of~$C_n$ satisfies, for~$(\R^3)^n\ni v\mapsto |C_n(v)|$ measurable,
\begin{equation*}
	\sup_{Q\in\mc{Q}_m}P_{Q}\del[1]{|C_n(V)|\geq 0.5\cdot \eta_n^{1-1/m}}\geq 1-2\delta - \del[1]{e/4}^{\eta_nn/2}.
\end{equation*}
\end{theorem}
Theorem~\ref{thm:NoAdapt} shows that any confidence set~$C_n$ that has uniform coverage of at least~$1-\delta$ absent and present contamination must have a diameter at least of order~$\eta_n^{1-\frac1m}$. Thus, even though no contamination is present, any confidence set with coverage~$1-\delta$ \emph{if} there had been contamination must have a diameter depending on~$\eta_n$. This finding is similar in spirit to the one for constructing honest confidence intervals in a location model under Huber contamination in~\cite{luo2026adaptive}. Theorem~\ref{thm:NoAdapt} does not rule out the existence of point estimators that adapt to the unknown actual contamination fraction~$|\mc{O}|/n$.
 
\begin{remark}[Impossibility of estimating the actual contamination fraction] The same construction as used in the proofs of Theorems~\ref{thm:contamLB} and~\ref{thm:NoAdapt} can be used to show that is not possible to estimate the actual contamination fraction~$\rho(V):=|\mc{O}(V)|/n$ in~\eqref{eq:contamfrac} at precision~$o(\eta_n)$. Thus, one cannot hope to replace~$\eta_n$ in the implementation of~$\hat{\beta}_{n\text{W-2SLS}}$ (cf.,~\eqref{eq:epsfam}) by an estimator of~$\rho(V)$. 
	
\end{remark}
\begin{remark}[Huber contamination]
Consider estimating the mean~$\mu$ of a distribution~$P$ on~$\R$ satisfying~$\int |x-\mu|^m P(dx)\leq \sigma_m^m$ for some~$\sigma_m\in(0,\infty)$. When observing draws from the Huber contamination model~$M=(1-\eta)P+\eta \mathsf{L}$ for some contaminating distribution~$\mathsf{L}$ and~$\eta\in(0,1)$, the minimax rate is~$\eta^{1-\frac1m}+n^{-1/2}$, cf.~\cite{chen2016general}. Thus, by Theorems~\ref{thm:estim} and~\eqref{eq:LB} of Theorem~\ref{thm:contamLB}, there is no loss in terms of the minimax rates of convergence from allowing the more flexible data-dependent adversarial contamination instead of contamination from a fixed~$\mathsf{L}$ in a Huber contamination model.	
\end{remark}

\section{Inference robust to contamination and weak identification}
The positive results on estimation and inference in Theorems~\ref{thm:estim}--\ref{thm:FeasibleInference} imposed that the instruments are strong in the sense of~$\text{rank}[\E(z_1x_1')]=K$. We now show that without any assumptions on the smallest singular value of~$\E(z_1x_1')$, a natural winsorized version of the classic Anderson-Rubin test (cf.~\cite{anderson1949estimation}) is simultaneously robust to adversarial contamination \emph{and} weak instruments. Consider testing whether the ``structural'' parameter~$\beta$ in~\eqref{eq:multivariate} equals a candidate~$\beta_0$, that is testing
\begin{equation}\label{eq:ARhyp}
	\mathsf{H}_0:\beta=\beta_0\qquad\text{against}\qquad \mathsf{H}_1:\beta\neq \beta_0.
\end{equation}
Writing~$g_i(b)=z_i(y_i-x_i'b)$ for~$i=1,\hdots,n$ and~$b$ being~$K\times 1$, the classic Anderson-Rubin test based on the clean data utilises that~$\beta$ satisfies~$\E g_1(\beta)=\E(z_1(y_1-x_1'\beta))=\E(z_1u_1)=0$. If~$\beta_0$ satisfies~$\E g_1(\beta_0)=0$, the covariance matrix of $g_1(\beta_0)$ is~$\Sigma(\beta_0)=\E[ g_1(\beta_0)g_1(\beta_0)']$. Let~$\hat{\Sigma}_{\text{AR}}(\beta_0)=n^{-1}\sum_{i=1}^ng_i(\beta_0)g_i(\beta_0)'$ be its usual heteroskedasticity-robust sample-analogue and~$\bar{g}_n(\beta_0)=n^{-1}\sum_{i=1}^ng_i(\beta_0)$. The resulting heteroskedasticity-robust Anderson-Rubin statistic for~$\mathsf{H}_0:\beta=\beta_0$ is
\begin{equation}\label{eq:cleanARstat}
	\mathsf{AR}_n(\beta_0)
	=
	n\bar{g}_n(\beta_0)'\hat{\Sigma}_{\text{AR}}^{-1}(\beta_0)\bar{g}_n(\beta_0)
	=
	n^{-1}(y-X\beta_0)'Z\hat{\Sigma}_{\text{AR}}^{-1}(\beta_0)Z'(y-X\beta_0).
\end{equation}
Under~$\E g_1(\beta_0)=0$ and standard regularity conditions, $\mathsf{AR}_n(\beta_0)$  converges in distribution to a~$\chi^2(L)$ without imposing any assumptions on~$\text{rank}[\E(z_1x_1')]$. 

Although robust to weak instruments \emph{absent} contamination, $\mathsf{AR}_n(\beta_0)$ is extremely fragile to adversarial contamination. The following winsorized AR-test robustifies~$\mathsf{AR}_n(\beta_0)$ to adversarial contamination \emph{for free} in the same fashion as~$\hat{\beta}_{n,\text{W-2SLS}}$ robustified~$\hat{\beta}_{n,\text{2SLS}}$ for free in Section~\ref{sec:WinsMain} by replacing the same averages involved in~$\bar{g}_n(\beta_0)$ and~$\hat{\Sigma}_n$ by suitably winsorized means. To be precise, let~$\tilde{g}_i(b)=\tilde{z}_i(\tilde{y}_i-\tilde{x}_i'b)$ for~$i=1,\hdots,n$ and let~$\bar{\tilde{g}}_{n}(b)$ be the~$L\times 1$ vector with entries
\begin{equation*}
 \bar{\tilde{g}}_{n,l}(b)=\mu_{\eps_n}\del[1]{\tilde{g}_{1,l}(b),\hdots,\tilde{g}_{n,l}(b)},\qquad l=1,\hdots,L.
\end{equation*}
Furthermore, recalling the notation introduced just prior to~\eqref{eq:Xiestim}, let~$\hat{\Sigma}_{\text{W-AR}}(b)$ be the~$L\times L$ positive semi-definite matrix with entries
\begin{equation*}
\hat{\Sigma}_{\text{W-AR},l,j}(b)
=
\frac{1}{n}\sum_{i=1}^n\phi_i\del[1]{\tilde{g}_{1,l}(b),\hdots, \tilde{g}_{n,l}(b)}	\cdot \phi_i\del[1]{\tilde{g}_{1,j}(b),\hdots, \tilde{g}_{n,j}(b)},\qquad 1\leq l,j\leq L.
\end{equation*}
We now define the following winsorized AR-statistic, which is simultaneously robust to i) adversarial contamination, ii) weak identification, and iii) heteroskedasticity:
\begin{equation*}
	\mathsf{AR}_{n,\text{W}}(\beta_0)
	=
	n\bar{\tilde{g}}_n(\beta_0)'\hat{\Sigma}_{\text{W-AR}}^{-1}(\beta_0)\bar{\tilde{g}}_n(\beta_0).\footnote{On the event that~$\hat{\Sigma}_{\text{W-AR}}(b)$ is not invertible, set~$\mathsf{AR}_{n,\text{W}}(b)=\infty$. This event has probability converging to zero under the assumptions of Theorem~\ref{thm:ARrobust}, but a fixed convention is crucial for the inversion in \eqref{eq:ARConf}.}
\end{equation*}
Because Assumption~\ref{ass:sequence} imposes the strong instrument assumption~$\text{rank}[\E(z_1x_1')]=K$, which is precisely the condition this section removes, Theorem~\ref{thm:ARrobust} is stated under primitive assumptions on the moment vector~$g_1(\beta_0)$ itself.
\begin{theorem}\label{thm:ARrobust}
Fix~$\beta_0\in\R^K$. Suppose that~$\cbr[0]{y_i,x_i,z_i}_{i\in\N}$ is i.i.d.~and~$\E|g_{1,l}(\beta_0)|^m<\infty$ for ~$l=1,\hdots,L$ and some~$m>2$. Furthermore, suppose that $\E(g_1(\beta_0))=0$ and $\text{rank}\sbr[0]{\Sigma(\beta_0)}=L$. If $\eps_n$ is chosen as in~\eqref{eq:epsfam} and \eqref{eq:contamfrac} is satisfied with~$\sqrt{n}\eta_n^{1-\frac1m}\to 0$ as~$n\to\infty$, then
	\begin{equation}\label{eq:ARconv}
	\mathsf{AR}_{n,\text{W}}(\beta_0)\qquad \text{converges in distribution to }\chi^2(L).	
	\end{equation}
\end{theorem}
Theorem~\ref{thm:ARrobust} shows that the robustified Anderson-Rubin statistic~$\mathsf{AR}_{n,\text{W}}(\beta_0)$ --- taking the \emph{contaminated} data as input --- has the same limiting distribution as the classic heteroskedasticity-robust AR statistic~$\mathsf{AR}_{n}(\beta_0)$ in~\eqref{eq:cleanARstat} taking the \emph{clean} data as input. Thus, in this sense robustness can again be obtained for free also without imposing any assumptions of the strength of the instruments. We stress that Theorem~\ref{thm:ARrobust} establishes the asymptotic distribution of~$\mathsf{AR}_{n,\text{W}}(\beta_0)$ in the presence of contamination and thus goes beyond the clean-model asymptotic result for the robust AR statistic studied in \cite{klooster2024outlier}.

By Theorem~\ref{thm:ARrobust}, for any~$\delta\in(0,1)$, a confidence set with asymptotic coverage~$(1-\delta)$ can be constructed by standard test inversion as
\begin{equation}\label{eq:ARConf}
	\hat{C}_{\text{W-AR},1-\delta}
	=
	\cbr[1]{b\in\R^K: \mathsf{AR}_{n,\text{W}}(b)\leq \mathsf{k}_{1-\delta}(L)}, 
\end{equation}
where~$\mathsf{k}_{1-\delta}(L)$ denotes the~$(1-\delta)$-percentile of the~$\chi^2(L)$-distribution.

\section{A uniform finite-sample concentration gap: 2SLS vs.~W-2SLS }\label{sec:FiniteDev}
The previous sections have demonstrated through an \emph{asymptotic} lens how the winsorized 2SLS estimator robustifies the classic 2SLS estimator for free against adversarial contamination. We next show for fixed~$n$ that even absent contamination, amounting to~$\eta_n=0$, the deviations of~$\hat{\beta}_{n,\text{2SLS}}$ from the true~$\beta$ can be much larger than those of~$\hat{\beta}_{n,\text{W-2SLS}}$. Put differently, not only does the robustness of~$\hat{\beta}_{n,\text{W-2SLS}}$ come for free, it can also concentrate stronger around~$\beta$ than~$\hat{\beta}_{n,\text{2SLS}}$ even when the data is clean.  For the purpose of this finite-sample deviation study, we consider the just-identified case of~$L=K=1$ so that~$\hat{\beta}_{n,\text{2SLS}}$ in~\eqref{eq:2SLSM} reduces to
\begin{equation*}
	\hat{\beta}_{n,2\text{SLS}}
	=
	\frac{n^{-1}\sum_{i=1}^nz_iy_i}{n^{-1}\sum_{i=1}^nz_ix_i}
	=
	\beta+\frac{n^{-1}\sum_{i=1}^nz_iu_i}{n^{-1}\sum_{i=1}^nz_ix_i},
\end{equation*}
provided the denominator is non-zero (and set~$\hat{\beta}_{n,2\text{SLS}}=0$, otherwise) . If
\begin{enumerate}
	\item $\cbr[0]{x_i, z_i, u_i}_{i=1}^\infty\in\R^3$ are i.i.d.~satisfying $y_i=x_i\beta+u_i$ and
	\item $\E|z_1x_1|<\infty$,~$0<\E(z_1^2u_1^2)<\infty$, $\E(z_1u_1)=0$, and~$\pi:=\E(z_1x_1)\neq 0$,
\end{enumerate}
it follows by standard arguments that
\begin{equation*}
	\sqrt{n}\del[1]{\hat{\beta}_{n,2\text{SLS}}-\beta}
\qquad\text{converges in distribution to}\qquad
	\mathsf{N}\del[2]{0,\frac{\E(z_1^2u_1^2)}{\pi^2}},\qquad \text{as }n\to\infty.
\end{equation*}
Therefore, for any~$\delta\in(0,1)$ it holds that
\begin{equation*}
	\P\del[3]{|\hat{\beta}_{n,2\text{SLS}}-\beta|> \frac{\Phi^{-1}(1-\delta/2)}{\sqrt{n}}\cdot \frac{\sqrt{\E(z_1^2u_1^2)}}{|\pi|}}\to \delta,
\end{equation*}
where~$\Phi(\cdot)$ is the cdf of~$\mathsf{N}(0,1)$. Thus, since~$\Phi^{-1}(1-\delta/2)\leq \sqrt{2\log(2/\delta)}$,
\begin{equation}\label{eq:asymguar}
		\lim_{n\to\infty}\P\del[3]{|\hat{\beta}_{n,2\text{SLS}}-\beta|> \sqrt{\frac{2\log(2/\delta)}{n}}\cdot\frac{\sqrt{\E(z_1^2u_1^2)}}{|\pi|}}\leq  \delta.
\end{equation}
Since the ``failure probability''~$\delta$ only enters via~$\sqrt{\log(2/\delta)}$, it is tempting to conclude on the basis of these \emph{pointwise} asymptotic considerations that even for very small~$\delta$,~$\hat{\beta}_{n,2\text{SLS}}$ is close to~$\beta$ with probability at least~$1-\delta$. However, this is only true in the above pointwise sense as we shall now detail. To do so, like in Section \ref{sec:Opt}, for any distribution~$Q$ of~$R_i=(y_i,x_i,z_i,u_i)$ on~$\R^4$, let~$P_Q$ be a measure on the underlying measurable space~$(\Omega,\mc{F})$ such that~$P_Q\circ \del[1]{R_1,\hdots,R_n}^{-1}=Q^{\otimes n}$,  cf.~also Footnote~\ref{fn:prodspace}. For any estimator~$T_n$ of~$\beta$, any family of distributions~$\mc{Q}$ on~$\R^4$, and~$\delta\in(0,1)$, define the \emph{uniform deviation radius} (of~$T_n$ over~$\mc{Q}$ at failure probability $\delta$)
\begin{equation*}
	\tau_n(T_n,\mathcal{Q},\delta):=\inf_{t\geq 0}\cbr[3]{\sup_{Q\in\mc{Q}}P_Q(|T_n-\beta(Q)|> t)\leq \delta},
\end{equation*}
where, as usual, we set~$\inf \emptyset=\infty$. For given~$n$ and failure probability~$\delta$,~$\tau_n(T_n,\mathcal{Q},\delta)$ is the worst-case upper~$(1-\delta)$-quantile of the absolute estimation error, or equivalently the smallest deterministic half-length~$t$ of an interval~$[T_n-t,T_n+t]$ containing~$\beta(Q)$ with probability at least~$1-\delta$ uniformly over~$\mathfrak{Q}$. Clearly, a smaller~$\tau_n(T_n,\mathcal{Q},\delta)$ is desirable since it implies stronger concentration of~$T_n$ around~$\beta(Q)$ uniformly over~$\mc{Q}$.

We now show that the winsorized 2SLS estimator can have a much smaller uniform deviation radius than the plain 2SLS estimator. Fix~$\beta\in\R$,~$\underline{\pi}>0$,~$K>0$, and denote by~$\overline{\mathfrak{Q}}(\beta,\underline{\pi},K)$ those distributions~$Q$ of~$R_i=(y_i,x_i,z_i,u_i)$ on~$\R^4$ satisfying
\begin{align*}
	&P_Q\del[0]{y_1=\beta x_1+u_1}=1,\quad E_{P_Q}(z_1^2x_1^2)\leq K,\quad E_{P_Q}(z_1^2u_1^2)\leq K,\\
	&E_{P_Q}(z_1u_1)=0,\quad
	 |E_{P_Q}(z_1x_1)|\geq \underline{\pi}. 
\end{align*}
We stress that this is a setting of strong instruments since~$\underline{\pi}>0$ is fixed. Furthermore, the relevant second moments are uniformly bounded by~$K$. Thus, the poor concentration of the 2SLS estimator exhibited in Theorem~\ref{lem:Bound1} below is \emph{not} an artefact of weak instruments or diverging moments. Indeed,~$\hat{\beta}_{n,\text{W-2SLS}}$ will be seen to exhibit much stronger concentration around~$\beta$ uniformly over~$\overline{\mathfrak{Q}}(\beta,\underline{\pi},K)$. For the purpose of the lower bounds on the uniform concentration radius of~$\hat{\beta}_{n,\text{2SLS}}$ it even suffices to consider the sub-family of distributions~$\mathfrak{Q}$ on~$\R^4$, each~$Q$ element of which satisfies:
\begin{align*}
	&P_Q\del[1]{y_1=x_1+u_1}=1,\quad E_{P_Q}(z_1^2x_1^2)= 2,\quad E_{P_Q}(z_1^2u_1^2)= 2,\\
	&E_{P_Q}(z_1u_1)=0,\quad
	 E_{P_Q}(z_1x_1)=1,\quad E_{P_Q}(u_1)=0, \quad E_{P_Q}(x_1u_1)=1. 
\end{align*}
Note that for~$\beta=1$,~$\underline{\pi}\leq 1$, and~$K\geq 2$, it holds that~$\mathfrak{Q}\subseteq\overline{\mathfrak{Q}}(\beta,\underline{\pi},K)$. We stress that for the definition of~$\mathfrak{Q}$ the specific values of the expectations, choice of~$\beta=1$, and~$\pi=\E(z_1x_1)=1$ are for concreteness only and other choices would simply change the multiplicative constant in Theorem~\ref{lem:Bound1} below.\footnote{Indeed, a more general construction is provided in Lemma \ref{lem:dist1} in the appendix.}

\begin{theorem}\label{lem:Bound1}
Fix~$n\geq 2$ and $\delta\in\del[1]{0,1/4}$. Then,
\begin{equation*}
	\tau_n(\hat{\beta}_{n,\text{2SLS}},\mathfrak{Q},\delta)\geq \frac{1}{2}\del[2]{\frac{\delta^{-1}}{n}}^{1/2}.
\end{equation*}
\end{theorem}
The dependence of~$\tau_n(\hat{\beta}_{n,\text{2SLS}},\mathfrak{Q},\delta)$ on~$n$ and~$\delta$ in Theorem \ref{lem:Bound1} is sharp up to universal constants since Lemma \ref{lem:2sls-upper-bound} in the appendix and $\mathfrak{Q}\subseteq \overline{\mathfrak{Q}}(1,1,2)$ imply that for~$n\delta\geq 16$
\begin{equation*}
\tau_n(\hat{\beta}_{n,\text{2SLS}},\mathfrak{Q},\delta)
\leq 	
4\del[2]{\frac{\delta^{-1}}{n}}^{1/2}.
\end{equation*}
It follows that for~$\delta=\delta_n$
\begin{equation}\label{eq:IFF}
\tau_n(\hat{\beta}_{n,\text{2SLS}},\mathfrak{Q},\delta_n)\to 0\qquad \Longleftrightarrow\qquad \delta_nn\to \infty.	
\end{equation}

We next show that replacing the empirical IV moments by
confidence-calibrated winsorized means yields a fundamentally
smaller uniform deviation radius. To provide a clean comparison with the results for~$\hat{\beta}_{n,\text{2SLS}}$ in Theorem~\ref{lem:Bound1}, which fixes the failure probability~$\delta$, we consider the following variation on~$\eps_n$ in~\eqref{eq:epsfam} when implementing the winsorized 2SLS estimator for a given~$\delta\in(0,1)$. Set
\begin{equation}\label{eq:epsfamDevAna}
\eps:=\frac{\log(12/\delta)}{n}.
\end{equation} 
The winsorized 2SLS estimator in~\eqref{eq:W2SLSM} then reduces to\footnote{We set $\hat{\beta}_{n,\text{W-2SLS}}^{(\delta)}=0$ on the event~$\mu_{\eps}\del[0]{z_1x_1,\hdots,z_nx_n}=0$, which is outside the high-probability event in Theorem~\ref{thm:W-2SLS}} the confidence calibrated
\begin{equation*}
	\hat{\beta}_{n,\text{W-2SLS}}^{(\delta)}
	=
	\frac{\mu_{\eps}\del[1]{z_1y_1,\hdots,z_ny_n}}{\mu_{\eps}\del[1]{z_1x_1,\hdots,z_nx_n}}.
\end{equation*}

Observe that since~$\eta_n=0$ throughout this section, choosing~$\delta=12n^{-\lambda}$ implies that $\hat{\beta}_{n,\text{W-2SLS}}^{(12n^{-\lambda})}$ reduces to~$\hat{\beta}_{n,\text{W-2SLS}}$ in~\eqref{eq:W2SLSM} because~$\eps$ in~\eqref{eq:epsfamDevAna} equals~$\eps_n$ in~\eqref{eq:epsfam}. 
\begin{theorem}\label{thm:W-2SLS}
Fix~$\beta\in\R$,~$\underline{\pi}>0$,~$K>0$,~$n\in\N$, and~$\delta\in(0,1)$. Suppose that
\begin{equation}\label{eq:epscondDev}
\del[1]{3+\sqrt{5}}\frac{\log(12/\delta)}{n}<1\qquad\text{and}\qquad \mathfrak{B}\sqrt{K\frac{\log(12/\delta)}{n}}\leq \frac{\underline{\pi}}{2}
\end{equation}
where~$\mathfrak{B}$ is the positive constant given in~\eqref{eq:Bconstant} of the appendix. If~$\eps$ is chosen as in~\eqref{eq:epsfamDevAna} and~\eqref{eq:epscondDev} is satisfied, then for a positive constant~$C$ depending only on~$\beta,\underline{\pi}$ and~$K$, it holds that
\begin{equation}\label{eqn:m2bound}
	\tau_n\del[1]{\hat{\beta}_{n,\text{W-2SLS}}^{(\delta)},\overline{\mathfrak{Q}}(\beta,\underline{\pi},K),\delta}
	\leq 
	C \cdot \sqrt{\frac{\log(12/\delta)}{n}}.
\end{equation}
\end{theorem}
Since~$\mathfrak{Q}\subseteq \overline{\mathfrak{Q}}(1,1,2)$ Theorems~\ref{lem:Bound1} and~\ref{thm:W-2SLS} compare $\hat{\beta}_{n,\text{2SLS}}$ and~$\hat{\beta}_{n,\text{W-2SLS}}^{(\delta)}$ on the same strongly identified, bounded second moments class. Although the lower and upper bounds on the uniform deviation radii of $\hat{\beta}_{n,\text{2SLS}}$ and $\hat{\beta}_{n,\text{W-2SLS}}^{(\delta)}$ have the same~$n^{-1/2}$-dependence on~$n$ for fixed~$\delta$, they depend very differently on~$\delta$: 
\begin{align*}
	&\tau_n(\hat{\beta}_{n,\text{2SLS}},\mathfrak{Q},\delta_n)\to 0\qquad \Longleftrightarrow\qquad \delta_nn\to \infty,\\
	&\tau_n\del[1]{\hat{\beta}_{n,\text{W-2SLS}}^{(\delta_n)},\mathfrak{Q},\delta_n}\to 0\qquad \text{if}\qquad \log(1/\delta_n)/n\to 0.
\end{align*}
 This contrast in the concentration properties of~$\hat{\beta}_{n,\text{2SLS}}$ and~${\hat{\beta}_{n,\text{W-2SLS}}^{(\delta)}}$ is particularly striking for small failure probabilities~$\delta$. Studying failure probabilities~$\delta=\delta_n$ that converge to zero as~$n\to\infty$ can be informative since it is reasonable to demand a lower failure probability~$\delta$ as the sample size increases. 
  
%
\begin{corollary}\label{cor:conc}
Under the assumptions of Theorems~\ref{lem:Bound1} and~\ref{thm:W-2SLS} it holds that
\begin{equation}\label{eq:smallratio}
	\frac{\tau_n\del[1]{\hat{\beta}_{n,\text{W-2SLS}}^{(\delta)},\mathfrak{Q},\delta}}{\tau_n(\hat{\beta}_{n,\text{2SLS}},\mathfrak{Q},\delta)}
	\leq
	2C\cdot \sqrt{\delta\log(12/\delta)}
\end{equation}
Furthermore, for every~$\lambda\in(0,\infty)$ and~$\delta_n=12n^{-\lambda}$ it holds for~$n$ sufficiently large that
\begin{equation}\label{eq:PoorConc}
\tau_{n,\text{2SLS}}:=\tau_n(\hat{\beta}_{n,\text{2SLS}},\mathfrak{Q},12n^{-\lambda})\geq \frac{1}{2\sqrt{12}}\cdot n^{(\lambda-1)/2}	
\end{equation}
and there exists a constant~$C$  such that
\begin{equation}\label{eq:GoodConc}
	\tau_{n,\text{W-2SLS}}:=\tau_n\del[1]{\hat{\beta}_{n,\text{W-2SLS}}^{(12n^{-\lambda})},\mathfrak{Q},12n^{-\lambda}}\leq C \sqrt{\frac{\lambda\log(n)}{n}}.
\end{equation}
\end{corollary}
Recall that $\delta=12n^{-\lambda}$ implies that~$\hat{\beta}_{n,\text{W-2SLS}}^{(12n^{-\lambda})}$ equals $\hat{\beta}_{n,\text{W-2SLS}}$ in~\eqref{eq:W2SLSM} and that the assumptions of Theorems~\ref{lem:Bound1} and~\ref{thm:W-2SLS} are eventually satisfied. Together with~\eqref{eq:IFF},~Corollary \ref{cor:conc} therefore reveals that the ratio of uniform deviation radii satisfies
\begin{equation*}
	\tau_{n,\text{W-2SLS}}/\tau_{n,\text{2SLS}}\lesssim \sqrt{\log(n)/n^\lambda}\to 0,\qquad \text{for every }\lambda\in(0,\infty)
\end{equation*}
as well as the existence of the following regimes:

\begin{center}
\begin{tabular}{ccc}
\toprule
	&$\hat{\beta}_{n,\text{2SLS}}$& $\hat{\beta}_{n,\text{W-2SLS}}^{(12n^{-\lambda})}$\\
\midrule
$\lambda\in(0,1)$ & $\tau_{n,\text{2SLS}}\to 0$ & $\tau_{n,\text{W-2SLS}}\to 0$\\
$\lambda=1$ & $\tau_{n,\text{2SLS}}\not\to 0$ & $\tau_{n,\text{W-2SLS}}\to 0$\\
$\lambda\in(1,\infty)$ & $\tau_{n,\text{2SLS}}\to\infty$ & $\tau_{n,\text{W-2SLS}}\to 0$\\
\bottomrule
\end{tabular}
\end{center}
%
Hence winsorized 2SLS has a vanishing uniform deviation radius at every polynomial confidence level $1-12n^{-\lambda}$, whereas ordinary 2SLS ceases to have a vanishing uniform $(1-12 n^{-\lambda})$-deviation radius at~$\lambda=1$.
\begin{remark}
	The uniform finite-sample result in Theorem~\ref{lem:Bound1} puts the pointwise (in~$Q\in\mathfrak{Q}$) asymptotic guarantee in~\eqref{eq:asymguar} into perspective: Although for fixed~$Q\in\mathfrak{Q}$ and~$n\to\infty$, it holds by~\eqref{eq:asymguar} along with~$\pi=E_{P_Q}(z_1x_1)=1$ and~$E_{P_Q}(z_1^2u_1^2)=2$  that,
\begin{equation}\label{eq:pointwisw2}
	\lim_{n\to\infty}P_Q\del[3]{|\hat{\beta}_{n,\text{2SLS}}-\beta(Q)|>2\sqrt{\frac{\log(2/\delta)}{n}}}\leq \delta,
\end{equation}
Theorem~\ref{lem:Bound1} implies, in particular, that for~$n\geq 2$ and~$\delta\in(0,1/4)$ there exists a~$Q'\in\mathfrak{Q}$ depending on~$n$ and~$\delta$ for which
\begin{equation*}
P_{Q'}\del[3]{|\hat{\beta}_{n,\text{2SLS}}-\beta(Q')|>\frac{1}{4}\del[2]{\frac{\delta^{-1}}{n}}^{1/2}}>\delta.
\end{equation*}
Since~$\del[1]{\frac{\delta^{-1}}{n}}^{1/2}\gg\sqrt{\frac{\log(2/\delta)}{n}}$ for~$\delta$ small it follows that the convergence in \eqref{eq:pointwisw2} is \emph{not} uniform over~$\mathfrak{Q}$ even though~$E_{P_Q}(z_1x_1)=1$,~$E_{P_Q}(z_1u_1)=0$ and~$E_{P_Q}(z_1^2u_1^2)=E_{P_Q}(z_1^2x_1^2)=2$ for all~$Q\in\mathfrak{Q}$ in
\begin{equation*}
	\hat{\beta}_{n,\text{2SLS}}-\beta
	=
	\frac{n^{-1}\sum_{i=1}^nz_iu_i}{n^{-1}\sum_{i=1}^nz_ix_i}.
\end{equation*}
Thus, the lower bound in Theorem~\ref{lem:Bound1} is not driven by weak instruments or diverging moments. The construction in its proof keeps the denominator~$n^{-1}\sum_{i=1}^nz_ix_i$ at its expectation of one with high probability and exploits the familiar poor concentration of the empirical mean in the numerator~$n^{-1}\sum_{i=1}^nz_iu_i$, cf.~also Example 5.49 in~\cite{romanosiegel} and Proposition 6.2 in \cite{catoni2012challenging}.
	\end{remark}
	
\begin{remark}
As expected, in light of the results in the previous sections, one can establish a version of Theorem~\ref{thm:W-2SLS} that remains valid also in the presence of contamination. Such a result would replace the current right-hand side inside the probability in~\eqref{eqn:m2bound} by $C \cdot \del[1]{\sqrt{\eta_n}+\sqrt{\log(12/\delta)/n}}$, for a different constant~$C$. 	
\end{remark}

\newpage
\begin{appendix}
\numberwithin{equation}{section}

\section{Proof of Proposition~\ref{prop:IVClose}}

Let~$A$,~$B$, and~$c$ be~$L\times K$,~$L\times L$, and~$L\times 1$ real matrices, respectively. Let
\begin{equation*}
	\mc{D}=\cbr[1]{(A,B,c):\text{det}(B)\neq 0\text{ and }\text{det}(A'B^{-1}A)\neq 0}\subseteq\R^{KL+L^2+L},
\end{equation*}
and noting that~$\mc{D}$ non-empty and open, define~$\Psi:\mc{D}\to \R^K$ as
\begin{equation*}
	\Psi\del[1]{A,B,c}
	=
	\del[1]{A'B^{-1}A}^{-1}A'B^{-1}c.
\end{equation*}
For~$Q_{ZX}=\E(z_1x_1')$,~$Q_{ZZ}=\E(z_1z_1')$,~$Q_{Zy}=\E(z_1y_1)$, observe  at~$Q:=\del[1]{Q_{ZX},Q_{ZZ},Q_{Zy}}\in\mc{D}$ by Condition 3.~of Assumption~\ref{ass:sequence} and that that~$\Psi$ is continuously differentiable on~$\mc{D}$. Thus, there exists a compact convex set~$\mc{K}\subseteq \mc{D}$ with~$Q$ in its interior such that~$\Psi$ has a bounded derivative on~$\mc{K}$. By the law of large numbers it holds with probability converging to one that
\begin{equation*}
	n^{-1}\del[1]{Z'X,Z'Z,Z'y}\qquad\text{is an element of } \mc{K}.
\end{equation*}
We next argue that
\begin{equation}\label{eq:close}
\enVert[2]{n^{-1}\del[1]{Z'X,Z'Z,Z'y}-\del[1]{W_{ZX},W_{ZZ},W_{Zy}}}
=
O_{\P}\del[2]{\eps_n^{1-\frac{1}{m}}}	
\end{equation}
for all vector norms~$||\cdot||$ on~$\R^{KL+L^2+L}$. To this end, it suffices to argue that
\begin{align}
	&\envert[2]{\frac{1}{n}(Z'X)_{l,k}-W_{ZX,l,k}}
	=
	O_{\P}\del[2]{\eps_n^{1-\frac{1}{m}}},\quad \envert[2]{\frac{1}{n}(Z'Z)_{l,j}-W_{ZZ,l,j}}
	=
	O_{\P}\del[2]{\eps_n^{1-\frac{1}{m}}},\label{eq:firsttwo}\\
	&\envert[2]{\frac{1}{n}(Z'y)_{l}-W_{Zy,l}}
	=
	O_{\P}\del[2]{\eps_n^{1-\frac{1}{m}}},\quad \text{for }l,j=1,\hdots,L\text{ and }k=1,\hdots, K,\label{eq:third}
\end{align}
which we do by verifying the conditions of Lemma~\ref{lem:wins_to_arithm_mean} applied with~$\lambda_1=1.01$,~$\lambda_2=1$,~$\delta=6 n^{-\lambda}\to 0$, and~$\eta=\eta_n$ such that~$\eps$ there equals~$\eps_n=1.01\cdot \eta_n +\lambda\cdot \frac{\log(n)}{n}$. With these choices, Condition~\eqref{eq:epscond} is eventually satisfied.  Note also that it follows from $\E (|z_{1,l}z_{1,j}|^m)<\infty$ and~$\E(|z_{1,l}x_{1,k}|^m)<\infty$ for~$l,j=1,\hdots,L$ and~$k=1,\hdots,K$ that
\begin{equation*}
	\E\envert[1]{z_{1,l}x_{1,k}-\E(z_{1,l}x_{1,k})}^m<\infty\qquad \text{and}\qquad \E\envert[1]{z_{1,l}z_{1,j}-\E(z_{1,l}z_{1,j})}^m<\infty,
\end{equation*}
such that Lemma~\ref{lem:wins_to_arithm_mean} applied with~$S_i=z_{i,l}x_{i,k}$ and~$S_i=z_{i,l}z_{i,j}$, respectively, for~$i=1,\hdots,n$ yields~\eqref{eq:firsttwo} (for a desired tail probability, first choose~$M$ in Lemma~\ref{lem:wins_to_arithm_mean} large, then~$n$ large). Finally, denoting by~$||\cdot||_\infty$ the supremum-norm of a vector with real entries, observe that
\begin{equation*}
	|z_{1,l}y_1|
	=
	|z_{1,l}x_{1}'\beta+z_{1,l}u_1|
	\leq
	\sum_{k=1}^K|z_{1,l}x_{1,k}|\cdot||\beta||_\infty+|z_{1,l}u_1|.
\end{equation*}
implying that
\begin{equation}\label{eq:BDmoments1}
	\E|z_{1,l}y_1|^m
	\leq
	(K+1)^m||\beta||_\infty^m\cdot\sum_{k=1}^K\E|z_{1,l}x_{1,k}|^m+(K+1)^m\E|z_{1,l}u_1|^m
	<\infty,
	\end{equation}
and hence
\begin{equation}\label{eq:BDmoments}
	\E|z_{1,l}y_1-\E(z_{1,l}y_1)|^m<\infty,\qquad \text{for}\qquad l=1,\hdots,L,
\end{equation}
so that we can also use Lemma~\ref{lem:wins_to_arithm_mean} with~$S_i=z_{i,l}y_i$ to conclude~\eqref{eq:third} and hence~\eqref{eq:close}. Therefore, because~$\eps_n\to 0$, also~$\del[1]{W_{ZX},W_{ZZ},W_{Zy}}\in \mc{K}$ with probability converging to one and so the boundedness of the derivative of~$\Psi$ on~$\mc{K}$ implies that by the mean value inequality
\begin{equation*}
	\hat{\beta}_{n,\text{W-2SLS}}-\hat{\beta}_{n,\text{2SLS}}
	=
	\Psi\del[2]{\del[1]{n^{-1}Z'X,n^{-1}Z'Z,n^{-1}Z'y}}-\Psi\del[1]{W_{ZX},W_{ZZ},W_{Zy}}
	=
	O_{\P}\del[2]{\eps_n^{1-\frac{1}{m}}}.	
\end{equation*}
\section{Proof of Theorems~\ref{thm:estim} and~\ref{thm:limdist}}
\begin{proof}[Proof of Theorem~\ref{thm:estim}]
Under the stated assumptions~$\hat{\beta}_{n,\text{2SLS}}$ based on the clean data satisfies~$\enVert[0]{\hat{\beta}_{n,\text{2SLS}}-\beta}=O_\P(n^{-1/2})$. Therefore, in combination with Proposition~\ref{prop:IVClose}, it follows that
\begin{align*}
\enVert[1]{\hat{\beta}_{n,\text{W-2SLS}}-\beta}
&\leq
\enVert[1]{\hat{\beta}_{n,\text{W-2SLS}}-\hat{\beta}_{n,\text{2SLS}}}+\enVert[1]{\hat{\beta}_{n,\text{2SLS}}-\beta}\\
&=
O_\P\del[3]{\eta_n^{1-\frac{1}{m}}+\sbr[2]{\frac{\log(n)}{n}}^{1-\frac{1}{m}}}+O_\P(n^{-1/2}).
	\end{align*} 
	Because~$m>2$, it holds that~$\sbr[1]{\frac{\log(n)}{n}}^{1-\frac{1}{m}}=O(n^{-1/2})$ from which the conclusion of the theorem follows.
\end{proof}

\begin{proof}[Proof of Theorem~\ref{thm:limdist}]
Under the stated assumptions 
\begin{equation*}
	\sqrt{n}\del[1]{\hat{\beta}_{n,\text{2SLS}}-\beta}\quad\text{converges in distribution to}\quad \mathsf{N}_K\del[1]{0,A\E(u_1^2z_1z_1')A'}.
\end{equation*}	
Since by Proposition~\ref{prop:IVClose} it also holds that
\begin{equation*}
	\enVert[1]{\sqrt{n}\del[1]{\hat{\beta}_{n,\text{W-2SLS}}-\hat{\beta}_{n,\text{2SLS}}}}
	=
	O_\P\del[3]{\sqrt{n}\eta_n^{1-\frac{1}{m}}+\frac{\log(n)^{1-\frac1m}}{n^{\frac12-\frac1m}}}
	=
	o_\P(1),
\end{equation*}
the last equality following from~$\sqrt{n}\eta_n^{1-\frac1m}\to 0$ and $\log(n)^{1-\frac{1}{m}}/n^{\frac{1}{2}-\frac{1}{m}}\to 0$ (recall that~$m>2$). Thus, also
\begin{equation*}
	\sqrt{n}\del[1]{\hat{\beta}_{n,\text{W-2SLS}}-\beta}\quad\text{converges in distribution to}\quad \mathsf{N}_K\del[1]{0,A\E(u_1^2z_1z_1')A'}.
\end{equation*}	
\end{proof}

\section{Proof of Theorem~\ref{thm:FeasibleInference}}
We begin by establishing that~$\hat{A}=\del[1]{W_{ZX}'[W_{ZZ}]^{-1}W_{ZX}}^{-1}W_{ZX}'[W_{ZZ}]^{-1}$ converges in probability to~$A$. To this end, note that it follows from the assumed~$\E (|z_{1,l}z_{1,j}|^m)<\infty$ and~$\E(|z_{1,l}x_{1,k}|^m)<\infty$ for~$l,j=1,\hdots,L$ and~$k=1,\hdots,K$ that
\begin{equation*}
	\E\envert[1]{z_{1,l}x_{1,k}-\E(z_{1,l}x_{1,k})}^m<\infty\qquad \text{and}\qquad \E\envert[1]{z_{1,l}z_{1,j}-\E(z_{1,l}z_{1,j})}^m<\infty.
\end{equation*}
Furthermore, by~\eqref{eq:BDmoments} it also holds that~$\E|z_{1,l}y_1-\E(z_{1,l}y_1)|^m<\infty$ for~$l=1,\hdots,L$. Therefore, applying Theorem 2.1 of~\cite{Wins1} with~$\lambda_1=1.01$,~$\lambda_2=1$,~$\delta=6 n^{-\lambda}\to 0$, and~$\eta=\eta_n$ such that~$\eps$ there equals~$\eps_n=1.01\cdot \eta_n +\lambda\cdot \frac{\log(n)}{n}\to 0$ (implying that (5) there is eventually satisfied) it follows that
\begin{equation*}
	\del[1]{W_{ZX},W_{ZZ},W_{Zy}}\qquad \text{converges in probability to}\qquad \del[1]{\E(z_1x_1'),\E(z_1z_1'),\E(z_1y_1)}.
\end{equation*}
Hence, since~$\text{rank}[\E(z_1z_1')]=L$ and~$\text{rank}[\E(z_1x_1')]=K$, it follows by the continuous mapping theorem that~$\hat{A}\to A=\del[1]{\E(z_1x_1')'\sbr[0]{\E(z_1z_1')}^{-1}\E(z_1x_1')}^{-1}\E(z_1x_1')'\sbr[0]{\E(z_1z_1')}^{-1}$ in probability. 

Next, we argue that
\begin{equation*}
	\hat{\Omega}_S
	=
	\mathfrak{P}(\hat{\beta}_{n,\text{W-2SLS}})\hat{\Xi}\mathfrak{P}(\hat{\beta}_{n,\text{W-2SLS}})'
	\to 
	\mathfrak{P}(\beta)\Xi \mathfrak{P}(\beta)'
	=
	\Omega_S\qquad \text{in probability}.
\end{equation*}
By continuity of~$b\mapsto \mathfrak{P}(b)$ and~$\hat{\beta}_{n,\text{W-2SLS}}\to \beta$ in probability (by Theorem~\ref{thm:estim}), it follows that~$\mathfrak{P}(\hat{\beta}_{n,\text{W-2SLS}})\to \mathfrak{P}(\beta)$ in probability. Concerning the convergence of~$\hat{\Xi}$, let~$\check{\Xi}$ be the~$L(K+1)\times L(K+1)$ matrix with entries
\begin{equation*}
\check{\Xi}_{j,k}=\mu_{\eps_n}\del[1]{\tilde{h}_{1,j}\tilde{h}_{1,k},\hdots,\tilde{h}_{n,j}\tilde{h}_{n,k}},\qquad\text{for}\qquad 1\leq j,k\leq L(K+1).	
\end{equation*}
Note that for every~$j\in\cbr[0]{1,\hdots,L(K+1)}$ either~$\tilde{h}_{i,j}=\tilde{z}_{i,l}\tilde{y}_i$ for some~$l\in\cbr[0]{1,\hdots,L}$ or~$\tilde{h}_{i,j}=\tilde{z}_{i,l}\tilde{x}_{i,k}$ for some~$l\in\cbr[0]{1,\hdots,L}$ and~$k\in\cbr[0]{1,\hdots,K}$. In either case, note that~$\E|z_{1,l}y_1|^m<\infty$, cf.~\eqref{eq:BDmoments1} in the proof of Proposition~\ref{prop:IVClose}, or~$\E|z_{1,l}x_{1,k}|^m<\infty$ by assumption. Hence,~$\E|h_{1,j}|^m<\infty$ for all~$1\leq j\leq L(K+1)$. We set up for an application of Lemma~\ref{lem:productwins} with~$S_{i,1}=h_{i,j}$,~$S_{i,2}=h_{i,k}$,~$\tilde{S}_{i,1}=\tilde{h}_{i,j}$, and~$\tilde{S}_{i,2}=\tilde{h}_{i,k}$. Applying Lemma~\ref{lem:productwins} with~$\lambda_1=1.01$,~$\lambda_2=1$,~$\delta=6n^{-\lambda}\to 0$ and~$\eta=\eta_n$ such that~$\eps$ there equals~$\eps_n=1.01\cdot \eta_n +\lambda\cdot \frac{\log(n)}{n}\to 0$ (noting that with these choices~\eqref{eq:epscond} is eventually satisfied) and~$m>2$ it follows that
\begin{equation*}
	\envert[1]{\hat\Xi_{j,k}-\check\Xi_{j,k}}\to 0\qquad \text{in probability for all} \qquad 1\leq j,k\leq L(K+1).
\end{equation*}
It therefore suffices to argue that~$\check\Xi$ converges to~$\Xi$ in probability. To this end, fix any pair~$1\leq j,k\leq L(K+1)$ and note that~$\E|h_{1,j}h_{1,k}|^{m/2}<\infty$ by the Cauchy-Schwarz inequality and the already established~$\E|h_{1,j}|^m<\infty$ for all~$1\leq j\leq L(K+1)$. Hence, also~$\E|h_{1,j}h_{1,k}-\E(h_{1,j}h_{1,k})|^{m/2}<\infty$. Thus, applying Theorem 2.1 of~\cite{Wins1} with~$X_i=h_{i,j}h_{i,k}$,~$\tilde{X}_i=\tilde{h}_{i,j}\tilde{h}_{i,k}$,~$\lambda_1=1.01$,~$\lambda_2=1$,~$\delta=6 n^{-\lambda}\to 0$, and~$\eta=\eta_n$ such that~$\eps$ there equals~$\eps_n=1.01\cdot \eta_n +\lambda\cdot \frac{\log(n)}{n}\to 0$ (implying that (5) there is eventually satisfied) and~$m$ there being~$m/2>1$ it follows that
\begin{equation*}
	\envert[1]{\check\Xi_{j,k}-\Xi_{j,k}}\to 0\qquad \text{in probability for all} \qquad 1\leq j,k\leq L(K+1).
\end{equation*}
Together with the penultimate display this implies that~$\hat{\Xi}\to \Xi$ and therefore~$\hat{\Omega}_S\to \Omega_S$ in probability. Therefore,
\begin{equation*}
	\hat{\Omega}
	=
	\hat{A}\hat{\Omega}_S\hat{A}'
	\to 
	A\Omega_SA'
	=
	\Omega\qquad\text{in probability},
\end{equation*}
implying also that~$\hat{\Omega}$ is positive definite with probability tending to one (since~$\Omega$ is positive definite under the assumptions of the theorem) and~$\hat{\Omega}^{-1/2}$ is eventually well-defined. Finally, as a result of this convergence and Theorem~\ref{thm:limdist} (note that we now impose that $\sqrt{n}\eta_n^{1-\frac1m}\to 0$),~\eqref{eq:feasible} follows by the continuous mapping theorem.

\section{Proofs of results in Section~\ref{sec:Opt}}
The following lemma constructs the distributions~$Q_0$ and~$Q_{1,n}$ in~$\mc{Q}_m$ used in the proofs of Theorems \ref{thm:contamLB} and \ref{thm:NoAdapt}. For any element~$a$ in a suitable space (typically~$\R^4$),~$\delta_{\cbr[0]{a}}$ denotes the Dirac measure at~$a$. 
\begin{lemma}\label{lem:distconstruct}
	Fix~$n\in\N$,~$\eta_n\in(0,1]$, and~$m\in[2,\infty)$. Let~$p_n=\eta_n/2$ and define the distributions
	\begin{align*}
		Q_0 &=\frac{1}{2}\delta_{\cbr[0]{-1,1,1,-1}}+\frac{1}{2}\delta_{\cbr[0]{1,1,1,1}}\\
		Q_{1,n}=Q_1 &= \frac{1}{2}\delta_{\cbr[0]{-1,1,1,-1-p_n^{1-\frac1m}}}+\del[2]{\frac{1}{2}-p_n}\delta_{\cbr[0]{1,1,1,1-p_n^{1-\frac1m}}}+p_n\delta_{\cbr[0]{1+p_n^{-\frac1m},1,1,1+p_n^{-\frac1m}-p_n^{1-\frac1m}}}
	\end{align*}
	with~$\beta(Q_0)=0$ and~$\beta(Q_1)=p_n^{1-\frac1m}$ of $(y_1,x_1,z_1,u_1)$. Then~$Q_0$ and~$Q_1$ belong to~$\mc{Q}_m$.
\end{lemma}		
\begin{proof}
That~$Q_0\in\mc{Q}_m$ is trivial since both of its atoms satisfy~$y_1=\beta(Q_0)x_1+u_1$ and the remaining conditions for membership of $\mc{Q}_m$ in~\eqref{eq:LBFam} are also easily seen to be satisfied. 	

To see that~$Q_1\in\mc{Q}_m$, we again observe that all atoms of~$Q_1$ satisfy $y_1=\beta(Q_1)x_1+u_1$ because~$\beta(Q_1)=p_n^{1-1/m}$. Because~$P_{Q_1}(z_1=1)=1$, it only remains to be shown that~$E_{P_{Q_1}}|u_1|^m\leq 1.5\cdot 2^m$,~$E_{P_{Q_1}}(u_1)=0$, and~$E_{P_{Q_1}}(u_1^2)\in[1,4]$. To this end, note that
\begin{align*}
	E_{P_{Q_1}}|u_1|^m
	&=
	0.5\cdot\del[1]{1+p_n^{1-1/m}}^m+(0.5-p_n)\cdot\envert[1]{1-p_n^{1-1/m}}^m+p_n\cdot\envert[1]{1+p_n^{-1/m}-p_n^{1-1/m}}^m\\
	&\leq
	0.5\cdot 2^m+(0.5-p_n)\cdot 2^m+p_n\cdot 2^{m-1}(1+p_n^{-1})\\
	&\leq
	1.5\cdot2^m,
\end{align*}
where we used that~$p_n\in(0,0.5]$ and~$|1+p_n^{-1/m}-p_n^{1-1/m}|^m\leq (1+p_n^{-1/m})^m\leq 2^{m-1}(1+p_n^{-1})$. Next,
\begin{align*}
	E_{P_{Q_1}}(u_1^2)
	&=
	0.5\cdot\del[2]{1+p_n^{1-\frac1m}}^2+(0.5-p_n)\cdot\envert[2]{1-p_n^{1-\frac1m}}^2+p_n\cdot\envert[2]{1+p_n^{-\frac1m}-p_n^{1-\frac1m}}^2\\
	&=
	1+2p_n^{1-\frac1m}+(1-p_n)p_n^{1-\frac2m}\in[1,4].
\end{align*}
and
\begin{equation*}
	E_{P_{Q_1}}(u_1)
	=
	0.5\cdot(-1-p_n^{1-1/m})+(0.5-p_n)\cdot (1-p_n^{1-1/m})+p_n\cdot (1+p_n^{-1/m}-p_n^{1-1/m})
	=0.
\end{equation*}
It follows that~$Q_1\in\mc{Q}_m$.
\end{proof}

For any two probability measures~$P_0$ and~$P_1$ on the same measurable space (typically~$(\R^3)^n,\mc{B}(\R^3)^n$), we denote by~$\mathsf{TV}(P_0,P_1)=\sup_{A}|P_0(A)-P_1(A)|$ the total variation distance between~$P_0$ and~$P_1$, the supremum being over all measurable~$A$. 

Recall the definition of the contamination map~$\mathsf{C}$ in~\eqref{eq:Contam} and the remaining notation introduced prior to Theorem \ref{thm:contamLB}. The following lemma bounds the total variation distance between the distributions of the contaminated data~$(\tilde{V}_1,\hdots,\tilde{V}_n)=\mathsf{C}(R)$ under~$P_{Q_0}$ and~$P_{Q_1}$ with~$Q_0$ and~$Q_1$ defined in the statement of Lemma~\ref{lem:distconstruct}. Recall that
\begin{equation*}
	P_{Q_i}\circ (\mathsf{C}\circ R)^{-1}=(P_{Q_i}\circ R^{-1})\circ \mathsf{C}^{-1}=Q_{i}^{\otimes n}\circ \mathsf{C}^{-1}\qquad \text{for}\qquad i\in\cbr[0]{0,1}.
\end{equation*}

\begin{lemma}\label{lem:TVDist}
Fix~$n\in\N$,~$\eta_n\in(0,1]$, and~$m\in[2,\infty)$. Let~$p_n=\eta_n/2$. Then, for~$Q_0$ and~$Q_1$ defined the statement of Lemma~\ref{lem:distconstruct}
\begin{equation}\label{eq:TVBound}
	\mathsf{TV}\del[1]{Q_0^{\otimes n}\circ \mathsf{C}^{-1},Q_1^{\otimes n}\circ \mathsf{C}^{-1}}
	\leq 
	\del[2]{\frac{e}{4}}^{\eta_nn/2}.
\end{equation}	
\end{lemma}
\begin{proof}
Write~$\mathsf{U}:(\R^4)^n\to(\R^3)^n$ for the function~$(\tilde{V}_1,\hdots,\tilde{V}_n)=\mathsf{U}(R)$ with\footnote{In the course of the proof of the present lemma, it is convenient to interpret~$R$ as the input to the function~$\mathsf{U}$ rather than random variables. A similar comment applies to~$V_i=(y_i,x_i,z_i)$ and~$\tilde{V}_i=(\tilde{y}_i,\tilde{x_i},\tilde{z_i})$.}
\begin{equation*}
\tilde{V}_i=(\tilde{y}_i,\tilde{x_i},\tilde{z_i})=\del[1]{y_i-p_n^{-1/m}\mathds{1}_{\cbr[0]{y_i=1+p_n^{-1/m}}},x_i,z_i},\qquad i=1,\hdots,n,
\end{equation*} 
which replaces all~$y_i$ that equal~$1+p_n^{-1/m}$ by~$1$ (recall that~$\mathsf{C}$ only replaces the first up to~$\eta_nn$ $y_i$ that equal~$1+p_n^{-1/m}$ by~$1$).

Since there are no~$y_i$ that equal~$1+p_n^{-1/m}$ under~$Q_0^{\otimes n}$ and since~$\mathsf{U}$ replaces all~$y_i$ that equal~$1+p_n^{-1/m}$ by~$1$, it follows that
\begin{equation*}
	Q_0^{\otimes n}\circ \mathsf{C}^{-1}=Q_0^{\otimes n}\circ \mathsf{U}^{-1}=Q_1^{\otimes n}\circ \mathsf{U}^{-1};
\end{equation*}
here the second equality used that under~$Q_1^{\otimes n}$ one has that~$\mathsf{U}$ moves the probability mass~$p_n$ at~$1+p_n^{-1/m}$ of the~$y_i$ to~$1$, which already had mass~$(0.5-p_n)$, brining the total mass at~$1$ to~$0.5$, the same as under~$Q_0^{\otimes n}$. Therefore,
\begin{equation*}
	\mathsf{TV}\del[1]{Q_0^{\otimes n}\circ \mathsf{C}^{-1},Q_1^{\otimes n}\circ \mathsf{C}^{-1}}
	=
	\mathsf{TV}\del[1]{Q_1^{\otimes n}\circ \mathsf{U}^{-1},Q_1^{\otimes n}\circ \mathsf{C}^{-1}}.
\end{equation*}
Next, define
\begin{equation*}
\mc{B}_n=\cbr[3]{\sum_{i=1}^n\mathds{1}\cbr[0]{y_i=1+p_n^{-1/m}}\leq \eta_nn}
=
\cbr[3]{\sum_{i=1}^n\mathds{1}\cbr[0]{y_i=1+p_n^{-1/m}}\leq 2p_nn}
\end{equation*}
and let~$A\subseteq (\R^3)^n$ be measurable. Then
\begin{align*}
	&Q_1^{\otimes n}( \mathsf{U}^{-1}(A))-Q_1^{\otimes n}(\mathsf{C}^{-1}(A))\\
	=
	&Q_1^{\otimes n}(\mathsf{U}^{-1}(A)\cap \mc{B}_n)+Q_1^{\otimes n}(\mathsf{U}^{-1}(A)\cap \mc{B}_n^c)-Q_1^{\otimes n}(\mathsf{C}^{-1}(A)\cap \mc{B}_n)-Q_1^{\otimes n}(\mathsf{C}^{-1}(A)\cap \mc{B}_n^c).
\end{align*}
Observe that~$\mathsf{U}^{-1}(A)\cap \mc{B}_n=\mathsf{C}^{-1}(A)\cap \mc{B}_n$ because~$\mathsf{U}(R)=\mathsf{C}(R)$ for~$R\in\mc{B}_n$ as there are at most~$\eta_nn$ $y_i$ that equal~$1+p_n^{-1/m}$ on $\mc{B}_n$. Thus, for every~$A\subseteq (\R^3)^n$,
\begin{equation*}
\envert[1]{Q_1^{\otimes n}( \mathsf{U}^{-1}(A))-Q_1^{\otimes n}(\mathsf{C}^{-1}(A))}
=
\envert[1]{Q_1^{\otimes n}(\mathsf{U}^{-1}(A)\cap \mc{B}_n^c)-Q_1^{\otimes n}(\mathsf{C}^{-1}(A)\cap \mc{B}_n^c)}
\leq
Q_1^{\otimes n}(\mc{B}_n^c).	
\end{equation*}
Since~$\sum_{i=1}^n\mathds{1}\cbr[1]{y_i=1+p_n^{-1/m}}$ is binomially distributed with success probability~$p_n$ under~$Q_1^{\otimes n}$, it follows from (5) in~\cite{hagerup1990guided} with~$\epsilon$ there being 1 and~$m$ there being~$p_n n$  that  
\begin{equation*}
		Q_1^{\otimes n}\del[1]{\mc{B}_n^c}
		\leq 
		\del[1]{e/4}^{p_nn}
		=
		\del[1]{e/4}^{\eta_nn/2},
			\end{equation*}
which yields~\eqref{eq:TVBound} as the right-hand side does not depend on~$A$.
\end{proof}

\begin{proof}[Proof of Theorem~\ref{thm:contamLB}]
 Fix any measurable function~$T_n:\R^{3n}\to\R$ and let~$Q_0$ and~$Q_1$ be the elements of~$\mc{Q}_m$ as in the statement of Lemma~\ref{lem:distconstruct}.
 
 We first prove~\eqref{eq:LB}. Writing $E_i=\cbr[0]{|T_n-\beta(Q_i)|\geq p_n^{1-1/m}/2}$ for~$i\in\cbr[0]{0,1}$ and recalling that $\beta(Q_1)-\beta(Q_0)=p_n^{1-1/m}$, it follows that~$E_0\cup E_1=(\R^3)^n$. This implies via Lemma~\ref{lem:TVDist} that
\begin{align*}
	1
	\leq 
	Q_0^{\otimes n}\circ \mathsf{C}^{-1}\del[0]{E_0}+Q_0^{\otimes n}\circ \mathsf{C}^{-1}\del[0]{E_1}
	\leq
Q_0^{\otimes n}\circ \mathsf{C}^{-1}\del[0]{E_0}+Q_1^{\otimes n}\circ \mathsf{C}^{-1}\del[0]{E_1}+\del[2]{\frac{e}{4}}^{\eta_nn/2}.
\end{align*}	
Therefore, a simple rearrangement yields that
\begin{equation*}
\max_{i\in\cbr[0]{0,1}}Q_i^{\otimes n}\del[2]{|T_n(\mathsf{C})-\beta(Q_i)|\geq p_n^{1-1/m}/2}
	\geq
	0.5\cdot\sbr[2]{1-\del[2]{\frac{e}{4}}^{\eta_nn/2}}
\end{equation*}
which, recalling that~$Q_i^{\otimes n}=P_{Q_i}\circ R^{-1}$,~$\mathsf{C}(R)=\tilde{V}$, and~$p_n=\eta_n/2$ implies~\eqref{eq:LB}.

Concerning~\eqref{eq:NoUnifConv}, fix any functional~$\sigma:\mc{Q}_m\to(0,\infty)$  and suppose that (if this convergence does not holds, we are done)
\begin{equation*}
	\sup_{t\in\R}\envert[1]{Q_0^{\otimes n}\circ \mathsf{C}^{-1}\del[1]{\sqrt{n}\sbr[1]{T_n-\beta(Q_0)}\leq t}-\Phi(t/\sigma(Q_0))}
	\to 0.
\end{equation*}
Because~$\limsup_{n\to\infty}\sqrt{n}\eta_n^{1-1/m}>0$ there exists a subsequence along which~$\sqrt{n}\eta_n^{1-1/m}\geq c>0$, say, and~$\eta_n n\to\infty$ as~$m>2$. Therefore, by Lemma~\ref{lem:TVDist} it also holds that
\begin{equation*}
	\sup_{t\in\R}\envert[1]{Q_1^{\otimes n}\circ \mathsf{C}^{-1}\del[1]{\sqrt{n}\sbr[1]{T_n-\beta(Q_0)}\leq t}-\Phi(t/\sigma(Q_0))}
	\to 0
\end{equation*}
and~$d_n:=\sqrt{n}(\beta(Q_1)-\beta(Q_0))=\sqrt{n}(\eta_n/2)^{1-1/m}\geq \frac{c}{2^{1-1/m}}=:c'\in(0,\infty)$ along this subsequence. Hence,
\begin{equation*}
		\sup_{t\in\R}\envert[1]{Q_1^{\otimes n}\circ \mathsf{C}^{-1}\del[1]{\sqrt{n}\sbr[1]{T_n-\beta(Q_1)}\leq t}-\Phi((t+d_n)/\sigma(Q_0))}\to 0.
\end{equation*}
Because~$d_n\geq c'>0$,
\begin{equation*}
	\sup_{t\in\R}\envert[1]{\Phi((t+d_n)/\sigma(Q_0))-\Phi(t/\sigma(Q_1))}
	\geq
	 \envert[1]{\Phi(d_n/\sigma(Q_0))-\Phi(0/\sigma(Q_1))}
	 \geq \Phi(c'/\sigma(Q_0))-0.5,
\end{equation*}
 the right-hand side being strictly positive, this establishes~\eqref{eq:NoUnifConv}.
\end{proof}

\begin{proof}[Proof of Theorem~\ref{thm:NoAdapt}]
Fix all quantities as in the statement of the theorem and let~$Q_0$ and~$Q_1$ be as in the statement of Lemma~\ref{lem:distconstruct}. By assumption

\begin{equation*}
Q_{1}^{\otimes n}\circ \mathsf{C}^{-1}\del[1]{\beta(Q_1)\in C_n}
=
P_{Q_1}\del[1]{\beta(Q_1)\in C_n(\tilde{V})}\geq 1-\delta.
\end{equation*}	
which in combination with Lemma~\ref{lem:TVDist} implies that  
\begin{equation*}
	Q_{0}^{\otimes n}\circ \mathsf{C}^{-1}\del[1]{\beta(Q_1)\in C_n}\geq 1-\delta-\del[2]{\frac{e}{4}}^{\eta_nn/2}.
\end{equation*} 	
Under~$Q_{0}^{\otimes n}$ it holds that~$\mathsf{C}$ is the projection mapping from~$(\R^{4})^n\to(\R^3)^n$ so that~$P_{Q_0}(\mathsf{C}(R)=V)=1$, it also holds by assumption that
\begin{equation*}
Q_{0}^{\otimes n}\circ \mathsf{C}^{-1}\del[1]{\beta(Q_0)\in C_n}
=
P_{Q_0}\del[1]{\beta(Q_0)\in C_n(\mathsf{C}(R))}
=
P_{Q_0}\del[1]{\beta(Q_0)\in C_n(V)}\geq 1-\delta.
\end{equation*}	
Therefore, since~$\beta(Q_1)-\beta(Q_0)=(\eta_n/2)^{1-\frac1m}\geq 0.5\cdot \eta_n^{1-\frac{1}{m}}$, it follows by the two previous displays that
\begin{equation*}
	Q_{0}^{\otimes n}\circ \mathsf{C}^{-1}\del[1]{|C_n|\geq 0.5\cdot\eta_n^{1-\frac1m}}
	\geq
	Q_{0}^{\otimes n}\circ \mathsf{C}^{-1}\del[1]{\cbr[0]{\beta(Q_0), \beta(Q_1)}\subseteq C_n}
	\geq
	1-2\delta-\del[2]{\frac{e}{4}}^{\eta_nn/2},
\end{equation*}
so that the conclusion follows from
\begin{equation*}
	P_{Q_0}\del[1]{|C_n(V)|\geq 0.5\cdot\eta_n^{1-\frac1m}}
	=
	P_{Q_0}\del[1]{|C_n(\mathsf{C}(R))|\geq 0.5\cdot\eta_n^{1-\frac1m}}
	=
	Q_{0}^{\otimes n}\circ \mathsf{C}^{-1}\del[1]{|C_n|\geq 0.5\cdot\eta_n^{1-\frac1m}}.
\end{equation*}
\end{proof}

\section{Proof of Theorem~\ref{thm:ARrobust}}
Fix~$\beta_0$ as in the statement of the theorem. Under the conditions imposed, it holds that
\begin{equation*}
	\sqrt{n}\bar{g}_n(\beta_0)=\frac{1}{\sqrt{n}}\sum_{i=1}^ng_{i}(\beta_0)\qquad \text{converges in distribution to}\qquad \mathsf{N}_L\del[1]{0,\Sigma(\beta_0)}.
\end{equation*}
We next argue that for each~$l\in\cbr[0]{1,\hdots,L}$ 
\begin{equation*}
	\sqrt{n}\envert[1]{\bar{\tilde{g}}_{n,l}(\beta_0)-\bar{g}_{n,l}(\beta_0)}=o_\P(1)
\end{equation*}
by means of Lemma~\ref{lem:wins_to_arithm_mean} applied with~$\lambda_1=1.01$,~$\lambda_2=1$,~$\delta=6 n^{-\lambda}\to 0$, and~$\eta=\eta_n$ such that~$\eps$ there equals~$\eps_n=1.01\cdot \eta_n +\lambda\cdot \frac{\log(n)}{n}$. With these choices, Condition~\eqref{eq:epscond} is eventually satisfied.  Note also that it follows from $\E (|g_{1,l}(\beta_0)|^m)<\infty$ that~$\E|g_{1,l}(\beta_0)-\E g_{1,l}(\beta_0)|^m<\infty$ such that Lemma~\ref{lem:wins_to_arithm_mean} applied with~$S_i=g_{i,l}(\beta_0)$ and~$\tilde{S}_i=\tilde{g}_{i,l}(\beta_0)$, for~$i=1,\hdots,n$ yields the previous display (for a desired tail probability, first choose~$M$ in Lemma~\ref{lem:wins_to_arithm_mean} large, then~$n$ large) noting also that~$\sqrt{n}\eps_n^{1-\frac1m}\to 0$. As a result, also
\begin{equation}\label{eq:convdistAR}
	\sqrt{n}\bar{\tilde{g}}_{n}(\beta_0)\qquad \text{converges in distribution to}\qquad \mathsf{N}_L\del[1]{0,\Sigma(\beta_0)}.
\end{equation}
We next argue that~$\hat{\Sigma}_{\text{W-AR}}(\beta_0)\to \Sigma(\beta_0)$ in probability by arguments similar to those already employed in the proof of Theorem~\ref{thm:FeasibleInference} (but which we adjust to the present setup for completeness): Let~$\check{\Sigma}(\beta_0)$ be the~$L\times L$ matrix with entries
\begin{equation*}
\check{\Sigma}_{l,j}(\beta_0)
=
\mu_{\eps_n}\del[1]{\tilde{g}_{1,l}(\beta_0)\tilde{g}_{1,j}(\beta_0),\hdots, \tilde{g}_{n,l}(\beta_0) \tilde{g}_{n,j}(\beta_0)},\qquad 1\leq l,j\leq L.
\end{equation*}
Note that~$\E|g_{1,l}(\beta_0)g_{1,j}(\beta_0)|^{m/2}<\infty$ by the Cauchy-Schwarz inequality. We set up for an application of Lemma~\ref{lem:productwins} with~$S_{i,1}=g_{i,l}(\beta_0)$,~$S_{i,2}=g_{i,j}(\beta_0)$,~$\tilde{S}_{i,1}=\tilde{g}_{i,l}(\beta_0)$, and~$\tilde{S}_{i,2}=\tilde{g}_{i,j}(\beta_0)$. Applying Lemma~\ref{lem:productwins} with~$\lambda_1=1.01$,~$\lambda_2=1$,~$\delta=6n^{-\lambda}\to 0$ and~$\eta=\eta_n$ such that~$\eps$ there equals~$\eps_n=1.01\cdot \eta_n +\lambda\cdot \frac{\log(n)}{n}\to 0$ (noting that with these choices~\eqref{eq:epscond} is eventually satisfied) and~$m>2$ it follows that
\begin{equation*}
	\envert[1]{\hat{\Sigma}_{\text{W-AR},l,j}(\beta_0)-\check{\Sigma}_{l,j}(\beta_0)}\to 0\qquad \text{in probability for all} \qquad 1\leq l,j\leq L.
\end{equation*}
It therefore suffices to argue that~$\check\Sigma(\beta_0)$ converges to~$\Sigma(\beta_0)$ in probability. To this end, fix any pair~$1\leq l,j\leq L$ and note that~$\E|g_{1,l}(\beta_0)g_{1,j}(\beta_0)-\E(g_{1,l}(\beta_0)g_{1,j}(\beta_0))|^{m/2}<\infty$ by the already established~$\E|g_{1,l}(\beta_0)g_{1,j}(\beta_0)|^{m/2}<\infty$. Thus, applying Theorem 2.1 of~\cite{Wins1} with~$X_i=g_{i,l}(\beta_0)g_{i,j}(\beta_0)$,~$\tilde{X}_i=\tilde{g}_{i,l}(\beta_0)\tilde{g}_{i,j}(\beta_0)$,~$\lambda_1=1.01$,~$\lambda_2=1$,~$\delta=6 n^{-\lambda}\to 0$, and~$\eta=\eta_n$ such that~$\eps$ there equals~$\eps_n=1.01\cdot \eta_n +\lambda\cdot \frac{\log(n)}{n}\to 0$ (implying that (5) there is eventually satisfied) and~$m$ there being~$m/2>1$ it follows that
\begin{equation*}
	\envert[1]{\check{\Sigma}_{l,j}(\beta_0)-\Sigma_{l,j}(\beta_0)}\to 0\qquad \text{in probability for all} \qquad 1\leq l,j\leq L.
\end{equation*}
Together with the penultimate display this implies that~$\hat{\Sigma}_{\text{W-AR}}(\beta_0)\to \Sigma(\beta_0)$ in probability, which in combination with~\eqref{eq:convdistAR} implies~\eqref{eq:ARconv}.

\section{Proofs of results in Section~\ref{sec:FiniteDev}}

\begin{lemma}\label{lem:dist1}
	Fix~$n\in \N$,~$\pi\in\R$,~$\beta\in\R$, $\sigma_C^2>0$,~$\sigma_V^2>0$, and~$a>\sigma_C/n$. Denote by~$\mu=\mu_V\otimes \mu_C\otimes \mu_Z$ the product measure on~$\R^3$ with
	\begin{itemize}
		\item $\mu_V=\frac{1}{2n^2}\delta_{-\sigma_Vn} +\frac{1}{2n^2}\delta_{\sigma_Vn}+\del[1]{1-\frac{1}{n^2}}\delta_0$;
		\item $\mu_C=\frac{\sigma_{C}^2}{2n^2a^2}\delta_{-an} +\frac{\sigma_{C}^2}{2n^2a^2}\delta_{an}+\del[1]{1-\frac{\sigma_{C}^2}{a^2n^2}}\delta_0$;
		\item $\mu_Z=\frac{1}{2}\delta_{-1}+\frac{1}{2}\delta_1$.
	\end{itemize}
	Let~$T:\R^3\to \R^4$ be the mapping defined via
	\begin{equation*}
		T(v,c,z)=\del[1]{\beta(\pi z+v)+zc+v,\pi z+v,z,zc+v}=\del[1]{T_i(v,c,z),\ i=1,\hdots,4}
	\end{equation*}
	and define the push-forward measure~$Q=\mu \circ T^{-1}$ on~$\R^4$. Let~$R_1=(y_1,x_1,z_1,u_1)$ have distribution~$Q$ and let~$P_Q\circ R_1^{-1}=Q$.\footnote{Confer Footnote~\ref{fn:prodspace} for regarding the existence of such a~$P_Q$.} Then
	\begin{align*}
		&P_Q(y_1=\beta x_1+u_1)=1,\quad E_{P_Q}(z_1u_1)=0,\quad E_{P_Q}(x_1u_1)=\sigma_V^2,\quad E_{P_Q}(x_1z_1)=\pi, \\
		 &E_{P_Q}(z_1^2x_1^2)=\pi^2+\sigma_V^2,\quad E_{P_Q}(z_1^2u_1^2)=\sigma_C^2+\sigma_V^2,\quad E_{P_Q}(u_1)=0.
	\end{align*}
	In particular, for~$\beta=\pi=1=\sigma_V^2=\sigma_C^2=1$ and stressing the dependence of~$Q=Q_a$ on~$a>1/n$, it holds that~$\cbr[0]{Q_a:a>1/n}\subseteq \mathfrak{Q}$.
	\end{lemma} 
	\begin{proof}
	Denote by~$V,\ C,$ and~$Z$ the coordinate projections on~$\R^3$ with~$\R^3$ equipped with the measure~$\mu=\mu_V\otimes \mu_C\otimes \mu_Z$, implying that, as random variables, $V,\ C,$ and~$Z$ are independent with~$V\sim \mu_V$,~$C\sim \mu_C$, and~$Z\sim \mu_Z$. Denote the expectation on~$\R^3$ by~$E_\mu$. 
	We shall use that for any $Q$-integrable~$f:\R^4\to \R$, it holds that
	\begin{align*}
		E_{P_Q}f(y_1,x_1,z_1,u_1)
		&=
		\int_{\R^4}f(y,x,z,u)dQ(y,x,z,u)\\
		&=
		\int_{\R^3}f(T(v,c,z))d\mu(v,c,z)\\
		&=
		E_{\mu}f(T(V,C,Z)).
	\end{align*}
 Then it holds that:
\begin{itemize}
\item $P_Q(y_1=\beta x_1+u_1)=Q\del[1]{(y,x,z,u)\in\R^4:y=\beta x+u}\\ =\mu\del[1]{(v,c,z)\in\R^3: T_1(v,c,z)=\beta T_2(v,c,z)+T_4(v,c,z)}=1$
\item $E_{P_Q}(z_1u_1)=E_{\mu}(Z(ZC+V))=E_{\mu}(C)+E_{\mu}(ZV)=0$; 
\item $E_{P_Q}(x_1u_1)=E_{\mu}(\pi Z+V)(ZC+V)=\pi E_{\mu}(C)+\pi E_{\mu}(ZV)+E_{\mu}(VZC)+E_{\mu}(V^2)=E_{\mu}(V^2)=\sigma_V^2$;
\item $E_{P_Q}(x_1z_1)=E_{\mu}((\pi Z+V)Z)=\pi+E_{\mu}(VZ)=\pi$;
\item $E_{P_Q}(z_1^2x_1^2)=E_{P_Q}(x_1^2)=E_{\mu}(\pi^2+V^2+2\pi ZV)=\pi^2+\sigma_V^2$;
\item $E_{P_Q}(z_1^2u_1^2)=E_{P_Q}(u_1^2)=E_{\mu}(C^2+V^2+2ZCV)=\sigma_C^2+\sigma_V^2$.
\item $E_{P_Q}(u_1)=E_{\mu}(ZC+V)=0$.
 \end{itemize}
Finally, for~$\beta=\pi=\sigma_v^2=\sigma_C^2=1$ it holds for each~$a>1/n$ that~$Q_a\in\mathfrak{Q}$.		
	\end{proof}

\begin{lemma}\label{lem:estimates}
	For~$n\geq 2$, it holds that
	\begin{equation}\label{eq:auxbounds}
		\frac{1}{4}
		\leq
		\frac{\del[1]{1-n^{-2}}^nn}{4(n-1)}\qquad\text{and}\qquad \frac{n}{4(n-1)}\leq \del[2]{1-\frac{1}{n}}^{n-1}.	
	\end{equation}
	Thus, for any~$p_n$ satisfying~$(1-n^{-2})^n\leq p_n$ it holds that~$1/4\leq p_n (1-n^{-1})^{n-1}$.  Furthermore, 
	\begin{equation}\label{eq:auxbound2}
		\del[1]{1-n^{-2}}^n\del[3]{1+\sqrt{1-\frac{(n-1)}{n\del[1]{1-n^{-2}}^n}}}\geq \frac{1}{2}.
	\end{equation}
	
\end{lemma}	

\begin{proof}
Concerning the first bound in~\eqref{eq:auxbounds}, Bernoulli's inequality gives
\begin{equation}\label{eq:Bern}
	\left(1-\frac{1}{n^2}\right)^n \ge 1-\frac{n}{n^2}=1-\frac1n.
\end{equation}
Hence,
\begin{equation*}
\left(1-\frac{1}{n^2}\right)^n \frac{n}{4(n-1)}
\ge \left(1-\frac1n\right)\frac{n}{4(n-1)}
= \frac14.	
\end{equation*}

For the second bound in~\eqref{eq:auxbounds}, let~$a_n:=\left(1-\frac1n\right)^n$. The sequence $(a_n)_{n\ge 2}$ is increasing and $a_2=1/4$, so $a_n\geq 1/4$ for all $n\geq 2$. Therefore,
\begin{equation*}
	\left(1-\frac1n\right)^{n-1}
=\frac{n}{n-1}\left(1-\frac1n\right)^n
\ge \frac{n}{n-1}\cdot \frac14
= \frac{n}{4(n-1)}.
\end{equation*}
Consequently, if $p_n\ge \left(1-n^{-2}\right)^n$, combining the just established two inequalities in~\eqref{eq:auxbounds} yields
\begin{equation*}
	\frac14\leq p_n\left(1-\frac1n\right)^{n-1}.
\end{equation*}

Concerning \eqref{eq:auxbound2}, by the first inequality in the already established~\eqref{eq:auxbounds},~$\frac{(n-1)}{n\del[0]{1-n^{-2}}^n}\leq 1$ so that the left-hand side of~\eqref{eq:auxbound2} is well-defined. Furthermore,~\eqref{eq:Bern} has already established that
\begin{equation*}
	q_n:=\left(1-\frac{1}{n^2}\right)^n\geq 1-\frac{1}{n}.
\end{equation*}
Thus, bounding the square root in~\eqref{eq:auxbound2} from below by zero, it follows that
\begin{equation*}
\del[1]{1-n^{-2}}^n\del[3]{1+\sqrt{1-\frac{(n-1)}{n\del[1]{1-n^{-2}}^n}}}
\geq
1-\frac{1}{n}
\geq 
\frac{1}{2},\qquad \text{for all }n\geq 2.
\end{equation*}
 
\end{proof}

\begin{proof}[Proof of Theorem~\ref{lem:Bound1}]
Fix~$n\geq 2$ and $\delta\in\del[0]{0,1/4}$. 

It suffices to consider the subfamily
$\{Q_a:a>1/n\}\subseteq\mathcal Q$ constructed in Lemma \ref{lem:dist1} and, as for~$Q=Q_a$ in the last line of that Lemma, we make explicit the dependence also of~$\mu=\mu_a=\mu_V\otimes \mu_{C,a}\otimes \mu_Z$ on~$a$ with~$\mu_V$,~$\mu_{C,a}=\mu_{C}$, and~$\mu_Z$ defined in the statement of that Lemma~\ref{lem:dist1}. Recall that~$\beta=\pi=\sigma_V^2=\sigma_C^2=1$. Furthermore, let~$T:\R^3\to\R^4$ be as defined there and denote by
\begin{equation*}
	(V_i,C_i,Z_i),\qquad i=1,\hdots,n 
\end{equation*}
the coordinate projections from the latent probability space~$((\R^3)^n, \mc{B}(\R^3)^n, \mu_a^n)$ to~$\R^3$. Finally, let~$(Y_i,X_i,Z_i,U_i):=T(V_i,C_i,Z_i)\in\R^4$ for~$i=1,\hdots,n$ such that
\begin{align*}
 Y_i=T_1(V_i,C_i,Z_i)=X_i+U_i,\qquad X_i=T_2(V_i,C_i,Z_i)=Z_i+V_i,\\ Z_i=T_3(V_i,C_i,Z_i)=Z_i,\qquad U_i=T_4(V_i,C_i,Z_i)=Z_iC_i+V_i. 	
\end{align*}			
Recalling that~$R_i=(y_i,x_i,z_i,u_i)$ and writing~$\hat{\beta}_{n,\text{2SLS}}=\hat{\beta}_{n,\text{2SLS}}(R_1,\hdots,R_n)$ it holds for all~$t\geq 0$ and~$a>1/n$ that (note that~$\beta(Q_a)=\beta=1$ for all~$a>1/n$.)
\begin{align}\label{eq:measurechange}
	P_{Q_a}\del[2]{\envert[1]{\hat{\beta}_{n,\text{2SLS}}((R_i)_{i=1}^n)-\beta}\geq t}
	&=
	Q_a^n\del[2]{\envert[1]{\hat{\beta}_{n,\text{2SLS}}((r_i)_{i=1}^n)-\beta}\geq t}\notag\\
	&=
	\mu_a^n\del[2]{\envert[1]{\hat{\beta}_{n,\text{2SLS}}((T(v_i,c_i,z_i))_{i=1}^n)-\beta}\geq t}\notag\\
	&=
	\mu_a^n\del[2]{\envert[1]{\hat{\beta}_{n,\text{2SLS}}((T(V_i,C_i,Z_i))_{i=1}^n)-\beta}\geq t}.
	\end{align}
Therefore, it suffices to consider the far right-hand side of the previous display concerning which we note that (whenever the denominator is non-zero)
\begin{align*}
\hat{\beta}_{n,\text{2SLS}}((T(V_i,C_i,Z_i))_{i=1}^n)
&=
\frac{n^{-1}\sum_{i=1}^nT_1(V_i,C_i,Z_i)T_3(V_i,C_i,Z_i)}{n^{-1}\sum_{i=1}^nT_2(V_i,C_i,Z_i)T_3(V_i,C_i,Z_i)}\\	
&=
\frac{n^{-1}\sum_{i=1}^n Y_iZ_i}{n^{-1}\sum_{i=1}^nX_iZ_i}\\
&=
1+\frac{n^{-1}\sum_{i=1}^n U_iZ_i}{n^{-1}\sum_{i=1}^nX_iZ_i},
\end{align*} 
the last equality using that~$Y_i=X_i+U_i$. Since also~$X_i=Z_i+V_i$ and $U_i=Z_iC_i+V_i$, 
\begin{equation*}
	\hat{\beta}_{n,\text{2SLS}}((T(V_i,C_i,Z_i))_{i=1}^n)
	=
	1+\frac{n^{-1}\sum_{i=1}^n (C_i+Z_iV_i)}{1+n^{-1}\sum_{i=1}^nZ_iV_i},
\end{equation*}
noting that~$\mu_a(Z_i^2=1)=1$ for all~$i=1,\hdots,n$. Therefore, since~$\beta=1$,
\begin{align*}
	\mu_a^n\del[1]{|\hat{\beta}_{n,\text{2SLS}}((T(V_i,C_i,Z_i))_{i=1}^n)-\beta|\geq t}
	=
	\mu_a^n\del[3]{\frac{|n^{-1}\sum_{i=1}^nC_i+n^{-1}\sum_{i=1}^nZ_iV_i|}{|1+n^{-1}\sum_{i=1}^nZ_iV_i|}\geq t}.
\end{align*}
Let~$\mathcal{G}_n=\cbr[1]{n^{-1}\sum_{i=1}^nZ_iV_i=0}$. Clearly~$\mc{G}_n$ occurs if~$V_1=\hdots=V_n=0$, implying that~$\mu_a^n(\mc{G}_n)\geq \del[1]{1-\frac{1}{n^2}}^n$. Thus,
\begin{align*}
	\mu_a^n\del[1]{|\hat{\beta}_{n,\text{2SLS}}((T(V_i,C_i,Z_i))_{i=1}^n)-\beta|\geq t}
	&\geq
		\mu_a^n\del[3]{\frac{|n^{-1}\sum_{i=1}^nC_i+n^{-1}\sum_{i=1}^nZ_iV_i|}{|1+n^{-1}\sum_{i=1}^nZ_iV_i|}\geq t, \mc{G}_n}\\
	&=
	\mu_a^n\del[3]{\envert[2]{n^{-1}\sum_{i=1}^nC_i}\geq t,\mathcal{G}_n}\\
	&=
	\mu_a^n\del[2]{\envert[1]{n^{-1}\sum_{i=1}^nC_i}\geq t}\cdot p_n,
\end{align*}
where the last equality used the independence of~$\mathcal{G}_n$ and~$\cbr[0]{C_i}_{i=1}^n$ as well as~$$p_n:=\mu_a^n(\mc{G}_n)=(\mu_V\otimes \mu_Z)^n(\mc{G}_n)\geq \del[1]{1-\frac{1}{n^2}}^n$$ not depending on~$a$. Now, for~$t=a$, it holds that~$|n^{-1}\sum_{i=1}^nC_i|\geq t$ (in particular) if one~$C_i$ is non-zero and the remaining~$n-1$ $C_i$ are zero. Thus, noting that~$\mu_a(C_1\neq 0)=\frac{1}{n^2a^2}\in(0,1)$ since~$a>1/n$, it holds that
\begin{align}\label{eq:reduction}
\mu_a^n\del[1]{|\hat{\beta}_{n,\text{2SLS}}((T(V_i,C_i,Z_i))_{i=1}^n)-\beta|\geq a}
\geq
n\cdot \frac{1}{n^2a^2}\del[2]{1-\frac{1}{n^2a^2}}^{n-1}\cdot p_n.
\end{align}
We now seek to set the right-hand side in the previous display equal to~$\delta\in(0,1)$ and solve for~$a$ to obtain a lower bound on the deviation of~$\hat{\beta}_{n,\text{2SLS}}((T(V_i,C_i,Z_i))_{i=1}^n)$ from~$\beta$ that occurs with probability at least~$\delta$. To this end, let~$s=\frac{1}{n^2a^2}$ and set~$f(s)=ns(1-s)^{n-1}\cdot p_n$ for~$s\in[0,1]$.~$f$ is maximized at~$n^{-1}$ with maximal value~$(1-n^{-1})^{n-1}\cdot p_n$. Since~$f$ is strictly increasing and continuous on~$[0,n^{-1}]$ there exists a unique solution~$s^*(\delta)\in[0,n^{-1})$ to~$f(s)=\delta\in (0,1/4)\subseteq (0,(1-n^{-1})^{n-1}\cdot p_n)$, the set inclusion following from~$(1-n^{-2})^n\leq p_n$ together with Lemma~\ref{lem:estimates}.\footnote{There is also a solution in~$[n^{-1},1]$ but since~$s=\frac{1}{n^2a^2}$, one has that smaller $s$ correspond to larger~$a$ and we are seeking to find the maximal~$a$ such that~$\hat{\beta}_{n,\text{2SLS}}((T(V_i,C_i,Z_i))_{i=1}^n)$ differs from~$\beta$ by~$a$ with probability at least~$\delta$. Hence, we are interested in small~$s$.} The corresponding value~$a^*(\delta)$ is

\begin{equation*}
	a^*(\delta)	
	=
	\del[2]{\frac{1}{n^2s^*(\delta)}}^{1/2}> \frac{1}{\sqrt{n}}\
> \frac1n. 
\end{equation*}
Therefore, 
\begin{equation*}
		\mu_{a^*(\delta)}^n\del[2]{\envert[1]{\hat{\beta}_{n,\text{2SLS}}((T(V_i,C_i,Z_i))_{i=1}^n)-\beta}\geq a^*(\delta)}
		\geq
		\delta,	
\end{equation*}
and we now argue that~$\sup_{a>1/n}\mu_{a}^n\del[1]{\envert[1]{\hat{\beta}_{n,\text{2SLS}}((T(V_i,C_i,Z_i))_{i=1}^n)-\beta}> r}>\delta$ for all~$r<a^*(\delta)$. Fix~$r<a^*(\delta)$. Because also~$a^*(\delta)>1/\sqrt{n}$ by the penultimate display, we can and do choose an~$a$ satisfying~$\max(r,1/\sqrt{n})<a<a^*(\delta)$. By~$r<a$ and~\eqref{eq:reduction},
\begin{align*}
	&\mu_a^n\del[1]{|\hat{\beta}_{n,\text{2SLS}}((T(V_i,C_i,Z_i))_{i=1}^n)-\beta|> r}\\
	\geq
	&\mu_a^n\del[1]{|\hat{\beta}_{n,\text{2SLS}}((T(V_i,C_i,Z_i))_{i=1}^n)-\beta|\geq a}
\geq
		f\del[2]{\frac{1}{n^2a^2}}>f\del[2]{\frac{1}{n^2a^*(\delta)^2}}
		=
		f(s^*(\delta))=\delta,	
\end{align*}
the strict inequality using that we have already argued that~$f$ is strictly increasing on~$[0,n^{-1}]$ and the final equality being by definition of~$s^*(\delta)$. Since~$r<a^*(\delta)$ was arbitrary, it follows that
\begin{equation*}
	\tau_n(\hat{\beta}_{n,\text{2SLS}},\mathfrak{Q},\delta)\geq a^*(\delta).
\end{equation*}

Finally, we upper bound~$s^*=s^*(\delta)$ in order to lower bound~$a^*(\delta)$ explicitly in terms of~$\delta$ for~$\delta\in(0,1/4)\subseteq(0,\frac{p_nn}{4(n-1)})$; the set inclusion following from Lemma~\ref{lem:estimates} for~$n\geq 2$ since~$(1-n^{-2})^n\leq p_n$. To this end, because~$(1-s)
^{n-1}\geq 1-(n-1)s$ for~$n\geq 2$ and~$s\in[0,1]$,
\begin{equation*}
\delta=
f(s^*)=ns^*(1-s^*)^{n-1}\cdot p_n
\geq
ns^*(1-(n-1)s^*)\cdot p_n.
\end{equation*}
The smallest root,~$r_-=r_-(\delta)$, say, of the quadratic equation~$ns(1-(n-1)s)p_n=\delta$ is
\begin{equation*}
r_-
	=
	\frac{1}{2(n-1)}\del[4]{1-\sqrt{1-\frac{4(n-1)\delta}{np_n}}}
	=
	\frac{\delta}{n}\Bigg\slash \sbr[4]{\frac{p_n}{2}\cdot\del[4]{1+\sqrt{1-\frac{4(n-1)\delta}{np_n}}}},
\end{equation*}
where the equality followed from multiplying and dividing the left-hand side by~$1+\sqrt{1-\frac{4(n-1)\delta}{np_n}}$ and using that~$\del[1]{1-\sqrt{1-\frac{4(n-1)\delta}{np_n}}}\del[1]{1+\sqrt{1-\frac{4(n-1)\delta}{np_n}}}=\frac{4(n-1)\delta}{np_n}$. From the first expression for~$r_-$ it is seen that~$r_-\leq n^{-1}$ for~$n\geq 2$. Thus, since
\begin{eqnarray*}
	f(r_-)=nr_-(1-r_-)^{n-1}p_n
	\geq
	 nr_-(1-(n-1)r_-)p_n=\delta
\end{eqnarray*}
and because~$f$ is increasing on~$[0,n^{-1}]$, it follows that~$s^*\leq r_-$. Hence, for~$\delta\in\del[1]{0,\frac{p_n n}{4(n-1)}}$,
\begin{align*}
	a^*(\delta)
	\geq
	\del[4]{\frac{1}{n}\cdot \frac{p_n\del[2]{1+\sqrt{1-\frac{4(n-1)\delta}{np_n}}}}{2\delta}}^{1/2}
	&=
	\del[2]{\frac{\delta^{-1}}{n}}^{1/2}\del[4]{\frac{p_n\del[2]{1+\sqrt{1-\frac{4(n-1)\delta}{np_n}}}}{2}}^{1/2}\\
	&\geq
	\frac{1}{2}\cdot\del[2]{\frac{\delta^{-1}}{n}}^{1/2},
	\end{align*}	
	the last inequality following from~$p_n\geq (1-n^{-2})^n$,~$\delta\leq 1/4$, and~\eqref{eq:auxbound2} of Lemma~\ref{lem:estimates}.
\end{proof}
	
\begin{proof}[Proof of Corollary~\ref{cor:conc}]
Since the assumptions of Theorems~\ref{lem:Bound1} and \ref{thm:W-2SLS} are imposed,~\eqref{eq:smallratio} is a direct consequence of these. Next, \eqref{eq:PoorConc} follows from Theorem~\ref{lem:Bound1} because~$\delta_n=12n^{-\lambda}\in(0,1/4)$, eventually. Similarly,~\eqref{eq:GoodConc} follows from Theorem \ref{thm:W-2SLS} since~$\delta_n=12n^{-\lambda}$ eventually satisfies the conditions in~\eqref{eq:epscondDev}.  	
\end{proof}	
		
The following upper bound on the uniform deviation radius of~$\hat{\beta}_{n,\mathrm{2SLS}}$ over~$\overline{\mathfrak{Q}}(\beta,\underline{\pi},K)$ complements the lower bound in Theorem \ref{lem:Bound1}	 over~$\mathfrak{Q}$. Recall that, e.g.,~$\mathfrak{Q}\subseteq \overline{\mathfrak{Q}}(1,1,2)$.
\begin{lemma}
\label{lem:2sls-upper-bound}
Fix $\beta\in\mathbb{R}$, $\underline{\pi}>0$, $K>0$,
$n\in\mathbb{N}$, $\delta\in(0,1)$, and suppose that
\begin{equation}\label{eq:2SLSdelta}
	n\underline{\pi}^2\delta
    \geq 8K,
\end{equation}
Then
\begin{equation*}
\tau_n\left(
    \widehat{\beta}_{n,\mathrm{2SLS}},
    \overline{\mathfrak{Q}}(\beta,\underline{\pi},K),
    \delta
\right)
\leq
\sqrt{\frac{8K}{\underline{\pi}^2\delta n}}
\leq 
1.	
\end{equation*}
\end{lemma}

\begin{proof}
Fix
$Q\in\overline{\mathfrak{Q}}(\beta,\underline{\pi},K)$, set $\pi_Q:=E_{P_Q}(z_1x_1)$ and abbreviate
\begin{equation*}
	    \widehat{\mu}_{n,zu}
    :=\frac{1}{n}\sum_{i=1}^n z_i u_i,
    \qquad
    \widehat{\mu}_{n,zx}
    :=\frac{1}{n}\sum_{i=1}^n z_i x_i.
\end{equation*}
Note that, whenever $\widehat{\mu}_{n,zx}\neq 0$,
\begin{equation*}
	\widehat{\beta}_{n,\mathrm{2SLS}}
=
\frac{n^{-1}\sum_{i=1}^n z_i y_i}
     {n^{-1}\sum_{i=1}^n z_i x_i}
=
\beta+
\frac{\widehat{\mu}_{n,zu}}
     {\widehat{\mu}_{n,zx}},
\end{equation*}
and by the definition of~$\overline{\mathfrak{Q}}(\beta,\underline{\pi},K)$,
\begin{equation*}
	  |\pi_Q|\geq\underline{\pi},
    \qquad
    \operatorname{Var}_{P_Q}(z_1u_1)
    =E_{P_Q}(z_1^2u_1^2)\leq K,\qquad \text{and}\qquad  \operatorname{Var}_{P_Q}(z_1x_1)\leq E_{P_Q}(z_1^2x_1^2)\leq K.
\end{equation*}
By Chebychev's inequality and~\eqref{eq:2SLSdelta} 
\begin{equation*}
P_Q\left(
    \left|\widehat{\mu}_{n,zx}-\pi_Q\right|
    \geq\frac{\underline{\pi}}{2}
\right)
\leq
\frac{4\operatorname{Var}_{P_Q}(z_1x_1)}
     {n\underline{\pi}^2} 
     \leq
 \frac{4K}
     {n\underline{\pi}^2}
     \leq 
     \frac{\delta}{2}.
\end{equation*}
Next, since $E_{P_Q}(z_1u_1)=0$, it holds for~$r_{n,\delta}=\sqrt{\frac{8K}{\underline{\pi}^2\delta n}}$ that
\begin{equation*}
P_Q\left(
    \left|\widehat{\mu}_{n,zu}\right|
    \geq\frac{r_{n,\delta}\underline{\pi}}{2}
\right)
\leq
\frac{4\operatorname{Var}_{P_Q}(z_1u_1)}
     {nr_{n,\delta}^2\underline{\pi}^2} 
\leq
\frac{4K}
     {nr_{n,\delta}^2\underline{\pi}^2} 
=
\frac{\delta}{2}.	
\end{equation*}
Hence, by the union bound, there exists an event $\mathcal{E}_{n}$
with $P_Q(\mathcal{E}_{n})\geq 1-\delta$ on which
\begin{equation*}
  \left|\widehat{\mu}_{n,zx}\right|
    \geq
    |\pi_Q|-
    \left|\widehat{\mu}_{n,zx}-\pi_Q\right| 
    >\frac{\underline{\pi}}{2}>0
    \qquad\text{and}\qquad
    \left|\widehat{\mu}_{n,zu}\right|
    <\frac{r_{n,\delta}\underline{\pi}}{2},	
\end{equation*}
so the 2SLS denominator~$\widehat{\mu}_{n,zx}$ is nonzero, and
\begin{equation*}
	\left|\widehat{\beta}_{n,\mathrm{2SLS}}-\beta\right|
=
\frac{\left|\widehat{\mu}_{n,zu}\right|}
     {\left|\widehat{\mu}_{n,zx}\right|} 
<
\frac{r_{n,\delta}\underline{\pi}/2}{\underline{\pi}/2}
=
r_{n,\delta}
\end{equation*}
It follows that
\begin{equation*}
	P_Q\left(
    \left|\widehat{\beta}_{n,\mathrm{2SLS}}-\beta\right|
    > r_{n,\delta}
\right)
\leq
P_Q\left(
    \left|\widehat{\beta}_{n,\mathrm{2SLS}}-\beta\right|
    \geq r_{n,\delta}
\right)
\leq 
\delta.
\end{equation*}
Since $Q$ was arbitrary,
\begin{equation*}
	\sup_{Q\in\overline{\mathfrak{Q}}(\beta,\underline{\pi},K)}
P_Q\left(
    \left|\widehat{\beta}_{n,\mathrm{2SLS}}-\beta\right|
    >r_{n,\delta}
\right)
\leq\delta.
\end{equation*}
The definition of the uniform deviation radius therefore gives
\begin{equation*}
	\tau_n\left(
    \widehat{\beta}_{n,\mathrm{2SLS}},
    \overline{\mathfrak{Q}}(\beta,\underline{\pi},K),
    \delta
\right)
\leq r_{n,\delta}
=
\sqrt{\frac{8K}{\underline{\pi}^2\delta n}}.
\end{equation*}
\end{proof}

\subsection{Proof of Theorem~\ref{thm:W-2SLS}}
Let
\begin{equation}\label{eq:Bconstant}
\mathfrak{B}:= \sqrt{2} +  \left( (A(\mathfrak{l}(1.01, 1), 1)/3)   +\overline{B}(1.01, 1) \right), 
\end{equation}
where
\begin{align*}
&A(d_1,d_2):=\frac{1}{d_1^{1/2}}+\frac{1}{d_2^{1/2}}\qquad\text{and}\qquad \overline{B}(1.01, 1):=2+\sbr[2]{1+\del[2]{\frac{\mathfrak{u}(1.01, 1)}{\mathfrak{l}(1.01, 1)}}^{\frac{1}{2}}}\mathfrak{u}(1.01, 1)^{\frac{1}{2}}\\
&\mathfrak{l}(1.01, 1) := (1-1.01^{-1})\exp \del[2]{{-\frac{1}{(1-1.01^{-1})}-1}}\quad\text{and}\quad  \mathfrak{u}(1.01, 1):=
	3+\sqrt{5}.
	\end{align*}

\begin{proof}[Proof of Theorem~\ref{thm:W-2SLS}]
Abbreviate~$\hat{\mu}_{n,zy}=\mu_{\eps}\del[1]{z_1y_1,\hdots,z_ny_n}$,~$\hat{\mu}_{n, zx}=\mu_{\eps}\del[1]{z_1x_1,\hdots,z_nx_n}$, and note that
\begin{equation}\label{eq:decomp}
	\hat{\beta}_{n,\text{W-2SLS}}^{(\delta)}
	=
	\frac{\hat{\mu}_{n,zy}}{\hat{\mu}_{n, zx}}
	=
	\beta 
	+\frac{\hat{\mu}_{n,zy}-\beta\hat{\mu}_{n, zx}}{\hat{\mu}_{n,zx}}.
\end{equation}	
Fix $Q\in \overline{\mathfrak{Q}}(\beta,\underline{\pi},K)$,~$\pi_Q:=E_{P_Q}(z_1x_1)$ and~$E_{P_Q}(z_1^2x_1^2)\leq K$. Note also that
\begin{align*}
	E_{P_Q}(z_1y_1)&=\beta E_{P_Q}(z_1x_1)+E_{P_Q}(z_1u_1)=\beta \pi_Q\qquad \text{and}\\
	E_{P_Q}(z_1^2y_1^2)&\leq 4\beta^2E_{P_Q}(z_1^2x_1^2)+4E_{P_Q}(z_1^2u_1^2)\leq 4K(\beta^2+1).
\end{align*}
We set up for two applications of Theorem 2.1 in~\cite{Wins1} with~$\eta$ there being zero (no contamination),~$\delta$ there replaced by~$\frac{1}{2}\delta$,~$\lambda_2=1$ such that~$\eps$ there equals the~$\eps$ in~\eqref{eq:epsfamDevAna} and (5) there is satisfied by~\eqref{eq:epscondDev}. Furthermore, for~$m=2$, and~$X_i$ there being~$z_iy_i$ and~$z_ix_i$, respectively, one has that by the union bound there exists a set~$E_n$ with probability at least~$1-\delta$ on which 
\begin{align*}
\envert[1]{\hat{\mu}_{n,zy}-\beta\pi_Q}
&\leq 
\mathfrak{B}\sqrt{4K(\beta^2+1)}\cdot \sqrt{\frac{\log(12/\delta)}{n}}
\leq
C_1\cdot \sqrt{\frac{\log(12/\delta)}{n}}\notag\qquad \text{and}\\	
\envert[1]{\hat{\mu}_{n,zx}-\pi_Q}
&\leq
\mathfrak{B}\sqrt{K}\cdot \sqrt{\frac{\log(12/\delta)}{n}}
\leq
C_1\cdot \sqrt{\frac{\log(12/\delta)}{n}},	
\end{align*}
for a constant~$C_1$ depending only on~$\beta$ and~$K$. Hence, on~$E_n$ it holds that
\begin{equation*}
	\envert[1]{\hat{\mu}_{n,zy}-\beta\hat{\mu}_{n, zx}}
	\leq
	\envert[1]{\hat{\mu}_{n,zy}-\pi_Q\beta}+|\beta|\envert[1]{(\hat{\mu}_{n,zx}-\pi_Q)}
	\leq
	C_2\cdot \sqrt{\frac{\log(12/\delta)}{n}},
\end{equation*}
for a constant~$C_2$ depending only on~$\beta$ and~$K$. Since by assumption 
\begin{equation*}
	|\hat{\mu}_{n,zx}|
	\geq
	 |\pi_Q|-\envert[0]{\hat{\mu}_{n,zx}-\pi_Q} 
	\geq
	 \underline{\pi}-\mathfrak{B}\sqrt{K\frac{\log(12/\delta)}{n}}
	 \geq
	 \frac{\underline{\pi}}{2},
\end{equation*}
the penultimate display together with~\eqref{eq:decomp} yields
\begin{equation*}
	P_Q\del[3]{|\hat{\beta}_{n,\text{W-2SLS}}^{(\delta)}-\beta |
	\leq 
	C\sqrt{\frac{\log(12/\delta)}{n}}}\geq 1-\delta,
\end{equation*}
for a constant~$C$ depending only on~$\beta,\underline{\pi}$ and~$K$. Because $Q\in \mathfrak{Q}(\beta,\underline{\pi},K)$ was arbitrary, the conclusion follows.

\end{proof}

\section{Auxiliary lemmas}
\subsection{Distance between winsorized mean of contaminated  and arithmetic mean uncontaminated observations}
For real numbers~$s_1,\hdots,s_n$, we denote by~$s_1^*\leq \hdots\leq s_n^*$ their non-decreasing rearrangement.
Let $S_1,\hdots,S_n$ and~$\tilde S_1,\hdots,\tilde S_n$ be real random variables. For~$\eps\in(0,1/2]$, let~$\hat{\alpha}=\tilde S_{\lceil \eps n \rceil}^*$ and $\hat{\beta}=\tilde S_{\lfloor(1-\eps )n\rfloor+1}^*$. Define the (one-dimensional) quantile-based winsorized mean estimator
\begin{equation}\label{eqn:winsmean}
\hat{\mu}_{n}
=
\hat{\mu}_{n}(\eps):=\frac{1}{n}\sum_{i=1}^n\phi_{\hat\alpha,\hat\beta}(\tilde{S}_i),
\end{equation}
as in Section 2 of \cite{Wins1}. Observe that with the notation in~\eqref{eq:winsfunc} and~$\eps_n$ there being~$\eps$, one can write
\begin{equation}\label{eq:notatinalequivalence}
	\hat{\mu}_n(\eps)
	=
	\mu_{\eps}(\tilde{S}_1,\hdots,\tilde{S}_n)
	=
	\frac{1}{n}\sum_{i=1}^{n}
\phi_{\tilde{S}_{(\lceil \eps n\rceil)}^*,\tilde{S}_{(\lfloor(1-\eps)n\rfloor)+1}^*}(\tilde{S}_i).
\end{equation}	

Lemma~\ref{lem:wins_to_arithm_mean} below controls the distance between the arithmetic mean $n^{-1}\sum_{i=1}^nS_i$ of $S_1,\hdots,S_n$ and the winsorized mean~$\hat{\mu}_{n}$ of~$\tilde S_1,\hdots,\tilde S_n$ in~\eqref{eqn:winsmean}. We next introduce some notation used in the proof of the lemma.

 For~$p\in(0,1)$ and a random variable~$S$, denote by~$Q_p(S)$ the $p$-quantile of the distribution of~$S$, that is
\begin{equation*}
Q_p(S)=\inf\cbr[1]{s\in \R:\P(S\leq s)\geq p}.
\end{equation*}
Furthermore, let
\begin{equation}\label{eq:epsApp}
	\eps=\lambda_1\cdot \eta +\lambda_2\cdot \frac{\log(6/\delta)}{n}
,\qquad \text{for fixed }\lambda_1\in(1,\infty)\text{ and }\lambda_2\in (0,\infty),
\end{equation}
and let~$A_+=1-\lambda_1^{-1}\mathds{1}\cbr[0]{\eta>0}$ as well as~$A_-=1+\lambda_1^{-1}\mathds{1}\cbr[0]{\eta>0}$. Define, for~$\eps$ as in~\eqref{eq:epsApp} and~$\delta \in (0, \infty)$, the quantities
\begin{equation}\label{eq:c1}
c_1 : = -A_+W_0\del[1]{-e^{-(\frac{\log(6/\delta)}{\eps n}+A_+)/A_+}}\in (0,A_+),
\end{equation}
as well as
\begin{equation}\label{eq:c2}
c_2 : = -A_-W_{-1}\del[1]{-e^{-(\frac{\log(6/\delta)}{\eps n}+A_-)/A_-}}\in(A_-,\infty),
\end{equation}
with~$W_0$ and~$W_{-1}$ being the principal and lower branch of Lambert's~$W$ function (cf., e.g., \cite{Corless1996}), respectively.

In fixed~$n$ results of Lemmas~\ref{lem:wins_to_arithm_mean} and~\ref{lem:productwins}, we shall impose that~$\eps$ in~\eqref{eq:epsApp} satisfies
\begin{equation}\label{eq:epscond}
	2\eps +\frac{\log(6/\delta)}{n}+\sqrt{\del[2]{\frac{\log(6/\delta)}{n}}^2+4\frac{\log(6/\delta)}{n}\eps}<1.
\end{equation}
This condition is typically satisfied for~$n\to\infty$.
\begin{lemma}\label{lem:wins_to_arithm_mean}
Fix~$n\in\N$,~$\delta\in(0,1)$,~$\eta\in[0,1]$, and~$M\geq 1$. Let~$S_1,\hdots,S_n$ be i.i.d.~random variables taking values in~$\R$ with~$\E |S_1|^m<\infty$ for some~$m\in[1,\infty)$,~$\mu:=\E S_1$, and~$\sigma_m^m:=\E|S_1-\mu|^m$. Furthermore, let~$\tilde{S}_1,\hdots,\tilde{S}_n$ be random variables in~$\R$ satisfying  
\begin{equation*}
	\envert[1]{\cbr[1]{i\in\cbr[0]{1,\hdots,n}:\tilde{S}_i\neq S_i}}\leq \eta n.
\end{equation*}	
Then for~$\hat{\mu}_{n}(\eps)$ as defined in~\eqref{eqn:winsmean} with~$\eps$ as in~\eqref{eq:epsApp} satisfying~\eqref{eq:epscond} it holds with probability at least~$1-\frac{4}{6}\delta-\frac{1}{M}$ that
\begin{equation}\label{eq:FDcontrol}
	\envert[3]{\hat{\mu}_{n}(\eps)-\frac{1}{n}\sum_{i=1}^n S_i}
	\leq
	CM\sigma_m\eps^{1-\frac{1}{m}}, 
\end{equation}
for a constant~$C=C(\lambda_1,\lambda_2,m)$.	
\end{lemma}
\begin{proof}
Fix all quantities as in the statement of the lemma and observe that
\begin{equation}\label{eq:lemdecomp}
	\envert[3]{\hat{\mu}_{n}(\eps)-\frac{1}{n}\sum_{i=1}^n S_i}
	\leq
	\envert[3]{\frac{1}{n}\sum_{i=1}^n\sbr[1]{\phi_{\hat\alpha,\hat\beta}(\tilde{S}_i)-\phi_{\hat\alpha,\hat\beta}(S_i)} }
	+
	\envert[3]{\frac{1}{n}\sum_{i=1}^n\sbr[1]{\phi_{\hat\alpha,\hat\beta}(S_i)-S_i}}.
\end{equation}
By Lemma B.5 in~\cite{Wins1} with~$\epsilon$ there being~$\eps$ (cf.~also Remark B.2 just after that lemma) and the union bound it holds on a set~$\mc{G}_n$ of probability at least~$1-\frac{4}{6}\delta $ that	
\begin{equation}\label{eq:orderstatbounds}
	Q_{c_1\eps}(S_1)\leq \hat \alpha \leq Q_{c_2\eps}(S_1)\qquad 
	\text{and}
	\qquad 
	Q_{1-c_2\eps}(S_1)\leq \hat \beta \leq Q_{1-c_1\eps}(S_1).
\end{equation}
Therefore, since~$\envert[1]{\cbr[1]{i\in\cbr[0]{1,\hdots,n}:\tilde{S}_i\neq S_i}}\leq \eta n$, it holds on~$\mc{G}_n$ that
\begin{equation*}
\envert[3]{\frac{1}{n}\sum_{i=1}^n\sbr[1]{\phi_{\hat\alpha,\hat\beta}(\tilde{S}_i)-\phi_{\hat\alpha,\hat\beta}(S_i)} }
\leq
\eta\del[1]{\hat{\beta}-\hat\alpha}	
\leq
\eta\del[1]{Q_{1-c_1\eps}(S_1)-Q_{c_1\eps}(S_1)}	
\leq
\frac{2\sigma_m\eta}{(c_1\eps)^{\frac{1}{m}}},
\end{equation*}
the last inequality following from Lemma C.1 in \cite{Wins1}. Furthermore, by Lemma B.3 in~\cite{Wins1} it follows that~$c_1$ (recall its definition in~\eqref{eq:c1}) is bounded from below by a positive constant depending only on~$\lambda_1$ and~$\lambda_2$. Using this, together with~$\eta<\eps $ (as~$\lambda_1>1$) in the previous display implies that 
\begin{equation}\label{eq:firsterm}
	\envert[3]{\frac{1}{n}\sum_{i=1}^n\sbr[1]{\phi_{\hat\alpha,\hat\beta}(\tilde{S}_i)-\phi_{\hat\alpha,\hat\beta}(S_i)} }
	\leq
	C_1\sigma_m\eps^{1-\frac{1}{m}},
\end{equation}
for a constant~$C_1$ depending only on~$\lambda_1,\lambda_2$, and~$m$.

Next, in order to bound the second term on the right-hand side of~\eqref{eq:lemdecomp}, note that on~$\mc{G}_n$ (cf.~\eqref{eq:orderstatbounds}) it also holds that for~$i=1,\hdots,n$ that
\begin{align}\label{eq:An}
&\envert[3]{\frac{1}{n}\sum_{i=1}^n\sbr[1]{\phi_{\hat\alpha,\hat\beta}(S_i)-S_i}}
\leq
\frac{1}{n}\sum_{i=1}^n\envert[1]{\phi_{\hat\alpha,\hat\beta}(S_i)-S_i}\notag\\
=&
\frac{1}{n}\sum_{i=1}^n\sbr[2]{(S_i-\hat\beta)\mathds{1}\cbr[0]{S_i>\hat\beta}+(\hat\alpha-S_i)\mathds{1}\cbr[0]{S_i<\hat\alpha}}\notag\\	
\leq &
\frac{1}{n}\sum_{i=1}^n\sbr[2]{(S_i-Q_{1-c_2\eps}(S_1))\mathds{1}\cbr[0]{S_i>Q_{1-c_2\eps}(S_1)}+(Q_{c_2\eps}(S_1)-S_i)\mathds{1}\cbr[0]{S_i<Q_{c_2\eps}(S_1)}}\notag\\
=:&
A_n. 
\end{align}
We proceed by bounding~$A_n$. Using that~$\P(S_1<Q_{c_2\eps}(S_1))\leq c_2\eps$ and~$\P(S_1>Q_{1-c_2\eps}(S_1))=1-\P(S_1\leq Q_{1-c_2\eps}(S_1))\leq c_2\eps$ along with H{\"o}lder's inequality (with the usual conventions in case~$m=1$), it holds that
\begin{align*}
&\E\sbr[1]{(S_1-Q_{1-c_2\eps}(S_1))\mathds{1}\cbr[0]{S_1>Q_{1-c_2\eps}(S_1)}}\\
\leq	
&\E|S_1-\mu|\mathds{1}\cbr[0]{S_1>Q_{1-c_2\eps}(S_1)}+(\mu-Q_{1-c_2\eps}(S_1))\P(S_1>Q_{1-c_2\eps}(S_1))\\
\leq
&\sigma_m (c_2\eps)^{1-\frac{1}{m}}+\frac{\sigma_m c_2\eps}{(1-c_2\eps)^{1/m}}\\
=
&\sigma_m \del[2]{1+\sbr[2]{\frac{c_2\eps}{1-c_2\eps}}^{1/m}}(c_2\eps)^{1-\frac{1}{m}},
\end{align*}
the penultimate inequality also using Lemma C.1 in~\cite{Wins1}. By identical arguments it also holds that
\begin{equation*}
	\E\sbr[1]{(Q_{c_2\eps}(S_1)-S_1)\mathds{1}\cbr[0]{S_1<Q_{c_2\eps}(S_1)}}
	\leq
	\sigma_m \del[2]{1+\sbr[2]{\frac{c_2\eps}{1-c_2\eps}}^{1/m}}(c_2\eps)^{1-\frac{1}{m}}.
\end{equation*}
Thus, the previous two displays imply that (recall that the~$S_i$ are identically distributed)
\begin{equation*}
\E A_n
\leq
2\sigma_m \del[2]{1+\sbr[2]{\frac{c_2\eps}{1-c_2\eps}}^{1/m}}(c_2\eps)^{1-\frac{1}{m}}.
\end{equation*}
By Lemma B.3 in \cite{Wins1} together with~$\eqref{eq:epscond}$ it follows that~$1-c_2\eps >c_1\eps$. Using also that that lemma bounds~$c_1$ and~$c_2$ from below and above by constants depending only on~$\lambda_1$ and~$\lambda_2$, the previous display implies the existence of a constant~$C_2$ depending only on~$\lambda_1,\lambda_2$, and~$m$ such that
\begin{equation*}
	\E A_n	
	\leq
	C_2\sigma_m \eps^{1-\frac{1}{m}}.
\end{equation*}
Therefore, by Markov's inequality it holds that
\begin{equation*}
	\P\del[1]{\mc{H}_n^c}
	\leq \frac{1}{M}, \qquad \text{where}\qquad \mc{H}_n=\cbr[2]{A_n\leq  M\cdot C_2\sigma_m \eps^{1-\frac{1}{m}}}.
\end{equation*} 
Thus, on~$\mc{G}_n\cap \mc{H}_n$ it holds by~\eqref{eq:An} and the previous display that
\begin{equation*}
\envert[3]{\frac{1}{n}\sum_{i=1}^n\sbr[1]{\phi_{\hat\alpha,\hat\beta}(S_i)-S_i}}
\leq
A_n
\leq 
M\cdot C_2\sigma_m \eps^{1-\frac{1}{m}}.	
\end{equation*}
Combining~\eqref{eq:lemdecomp},~\eqref{eq:firsterm}, and the previous display yields the conclusion of the lemma since~$M\geq 1$ and~$\P(\mc{G}_n\cap \mc{H}_n)\geq 1-\frac{4}{6}\delta-\frac{1}{M}$.
\end{proof}

\subsection{Joint vs.~coordinatewise Winsorization}
The following lemma compares the winsorized mean of a product with the average of the product of the separately winsorized coordinates. The proof of the lemma is similar to arguments used in the proof of Theorem 3.1 in~\cite{HDGauss}.
\begin{lemma}\label{lem:productwins}
	Fix~$n\in\N$,~$\delta\in(0,1)$, and~$\eta\in[0,1]$. Let~$(S_{i,1},S_{i,2}),\hdots,(S_{n,1},S_{n,2})$ be i.i.d.~random vectors taking values in~$\R^2$ with~$M:=\max_{j\in\cbr[0]{1,2}}\E|S_{1,j}|^m<\infty$ for some~$m\in[2,\infty)$. Furthermore, let~$(\tilde{S}_{1,1},\tilde{S}_{1,2}),\hdots,(\tilde{S}_{n,1},\tilde{S}_{n,2})$ be random vectors taking values in~$\R^2$ satisfying  
\begin{equation*}
	\envert[1]{\cbr[1]{i\in\cbr[0]{1,\hdots,n}:(\tilde{S}_{i,1}, \tilde{S}_{i,2})\neq (S_{i,1},S_{i,2})}}\leq \eta n.
\end{equation*}	
Then for~$\eps$ as in~\eqref{eq:epsApp} satisfying~\eqref{eq:epscond} it holds with probability at least~$1-\delta$ that
\begin{equation}\label{eq:winsprod}
	\envert[3]{\frac{1}{n}\sum_{i=1}^n\phi_{\hat{\alpha},\hat{\beta}}(\tilde{S}_{i,1}\tilde{S}_{i,2})-\frac{1}{n}\sum_{i=1}^n\phi_{\hat{\alpha}_1,\hat{\beta}_1}(\tilde{S}_{i,1})\phi_{\hat{\alpha}_2,\hat{\beta}_2}(\tilde{S}_{i,2})}
	\leq
	C\cdot M^{2/m}\eps^{1-\frac{2}{m}},
\end{equation}
for a constant~$C$ depending only on~$\lambda_1$,~$\lambda_2$, and~$m$ and where for~$\tilde{U}_i=\tilde{S}_{i,1}\tilde{S}_{i,2}$,~$i=1,\hdots,n$ the winsorization points are~$\hat{\alpha}=\tilde{U}_{\lceil \eps n\rceil}^*$,~$\hat{\beta}=\tilde{U}_{\lfloor (1-\eps) n\rfloor+1}^*$, $\hat{\alpha}_j=\tilde{S}_{\lceil \eps n\rceil,j}^*$, and~$\hat{\beta}_j=\tilde{S}_{\lfloor (1-\eps) n\rfloor+1,j}^*$ for~$j\in\cbr[0]{1,2}$.
\end{lemma}
\begin{proof}
To establish~\eqref{eq:winsprod}, note that with
\begin{align*}
A_{n}:=\cbr[2]{i\in\cbr[1]{1,\hdots,n}: \hat{\alpha}\leq \tilde{U}_{i}\leq \hat{\beta}_{} \text{ and }\hat{\alpha}_1\leq \tilde{S}_{i,1}\leq\hat{\beta}_1\text{ and }\hat{\alpha}_2\leq \tilde{S}_{i,2}\leq\hat{\beta}_2},
\end{align*}
one has that
\begin{align*}
\phi_{\hat{\alpha},\hat{\beta}}\del[0]{\tilde{U}_{i}}-\phi_{\hat{\alpha}_1,\hat{\beta}_1}(\tilde{S}_{i,1})\phi_{\hat{\alpha}_2,\hat{\beta}_2}(\tilde{S}_{i,2})
=
\tilde{U}_{i}-\tilde{S}_{i,1}\tilde{S}_{i,2}
=0\qquad\text{for }i\in A_{n}. 
\end{align*}
Hence, the left-hand side of~\eqref{eq:winsprod} equals
\begin{align*}
\mathfrak{S}_{n}:=\envert[3]{\frac{1}{n}\sum_{i\in A_{n}^c}\sbr[1]{\phi_{\hat{\alpha},\hat{\beta}}\del[0]{\tilde{U}_{i}}-\phi_{\hat{\alpha}_1,\hat{\beta}_1}(\tilde{S}_{i,1})\phi_{\hat{\alpha}_2,\hat{\beta}_2}(\tilde{S}_{i,2})
}},
\end{align*}	
and for every~$i\in A_{n}^c$ 
\begin{align}\label{eq:tailterms}
\envert[2]{\phi_{\hat{\alpha},\hat{\beta}}\del[0]{\tilde{U}_{i}}-\phi_{\hat{\alpha}_1,\hat{\beta}_1}(\tilde{S}_{i,1})\phi_{\hat{\alpha}_2,\hat{\beta}_2}(\tilde{S}_{i,2})}
\leq \del[1]{|\hat{\alpha}| \vee |\hat{\beta}|}	+\del[1]{|\hat{\alpha}_1| \vee |\hat{\beta}_1|}\del[1]{|\hat{\alpha}_2| \vee |\hat{\beta}_2|}.
\end{align}
To bound the right-hand side of~\eqref{eq:tailterms}, let~$U_i=S_{i,1}S_{i,2}$ and note that
\begin{equation*}
	\envert[1]{\cbr[1]{i\in\cbr[0]{1,\hdots,n}:\tilde{U}_{i}\neq U_{i}}}
	\leq
	\envert[1]{\cbr[1]{i\in\cbr[0]{1,\hdots,n}:(\tilde{S}_{i,1}, \tilde{S}_{i,2})\neq (S_{i,1},S_{i,2})}}\leq \eta n.
\end{equation*}
By Lemma B.5 in~\cite{Wins1} with~$\epsilon$ there being~$\eps$ (cf.~also Remark B.2 just after that lemma),~$X_i$ there being~$U_i$, and~$\tilde{X}_i$ there being~$\tilde{U}_i$, it holds that
\begin{align*}
Q_{c_{1}\eps}(U_{1})
\leq
\hat{\alpha}
\leq 
\hat{\beta}
\leq
Q_{1-c_1\eps}(U_{1})
\end{align*}
with probability at least~$1-\frac{\delta}{3}$, and where the second inequality used that~$\eps\in(0,1/2)$ by~\eqref{eq:epscond}. Therefore, with at least this probability,
\begin{align*}
|\hat{\alpha}| \vee |\hat{\beta}|
\leq
|Q_{c_{1}\eps}(U_{1})|\vee |Q_{1-c_{1}\eps}(U_{1})|.	
\end{align*}
By the same argument, it holds with probability at least~$1-\frac{2\delta}{3}$ that
\begin{equation*}
\del[1]{|\hat{\alpha}_1| \vee |\hat{\beta}_1|}\del[1]{|\hat{\alpha}_2| \vee |\hat{\beta}_2|}
\leq 
\del[1]{|Q_{c_{1}\eps}(S_{1,1})|\vee  |Q_{1-c_{1}\eps}(S_{1,1})|}
\cdot  \del[1]{|Q_{c_{1}\eps}(S_{1,2})|\vee  |Q_{1-c_{1}\eps}(S_{1,2})|}.
\end{equation*}
Thus, with probability at least~$1-\delta$ the right-hand side of~\eqref{eq:tailterms} is bounded from above by
\begin{equation*}
|Q_{c_{1}\eps}(U_{1})|\vee |Q_{1-c_{1}\eps}(U_{1})|	
+\del[1]{|Q_{c_{1}\eps}(S_{1,1})|\vee  |Q_{1-c_{1}\eps}(S_{1,1})|}
\cdot  \del[1]{|Q_{c_{1}\eps}(S_{1,2})|\vee  |Q_{1-c_{1}\eps}(S_{1,2})|}.
\end{equation*}
Next, note that
\begin{equation*}
	\E\envert[0]{U_{1}-\E U_{1}}^{m/2}
\leq 2^{m/2}\E|U_{1}|^{m/2}
=
2^{m/2}\E|S_{1,1}S_{1,2}|^{m/2}
\leq 
2^{m/2}M,
\end{equation*}
and~$\E |S_{1,j}-\E(S_{1,j})|^m\leq 2^mM$ for~$j=1,2$. Hence, Lemma~C.1 in~\cite{Wins1} implies that the penultimate display is bounded from above by 
\begin{align}\label{eq:intermediate}
\del[3]{|\E(U_1)|+\frac{2M^{2/m}}{(c_{1}\eps)^{2/m}}}+\del[3]{|\E S_{1,1}|+\frac{2M^{1/m}}{(c_{1}\eps)^{1/m}}}\cdot\del[3]{|\E S_{1,2}|+\frac{2M^{1/m}}{(c_{1}\eps)^{1/m}}}
\end{align}
Furthermore, because 
\begin{equation*}
|\E(U_1)|
=
	|\E(S_{1,1}S_{1,2})|
	\leq
	\del[2]{\E|S_{1,1}S_{1,2}|^{m/2}}^{2/m}
	\leq
	\del[2]{\sbr[1]{\E|S_{1,1}|^m\E|S_{1,2}|^m}^{1/2}}^{2/m}
	\leq
	M^{2/m}.
\end{equation*}
and
\begin{equation*}
|\E(S_{1,j})|
\leq
(\E|S_{1,j}|^m)^{1/m}
\leq	
M^{1/m},\qquad\text{for }j=1,2,	
\end{equation*}
~\eqref{eq:intermediate} is bounded from above by~$12M^{2/m}/(c_1\eps)^{2/m}$
Thus, since~$|A_{n}^c|\leq 6\eps n$, with probability at least~$1-\delta$ it holds that
\begin{align*}
\envert[1]{\mathfrak{S}_{n}}
\leq
6\eps\cdot \frac{12M^{2/m}}{(c_{1}\eps)^{2/m}}
\leq 
\frac{72M^{2/m}}{c_{1}^{2/m}}\cdot{\eps}^{1-\frac{2}{m}}
\leq 
C\cdot M^{2/m}\eps^{1-\frac{2}{m}},
\end{align*}
where we used the lower bound on~$c_{1}$ from Lemma~B.3 in~\cite{Wins1}, and~$C$ is the desired constant depending only on~$\lambda_1,\lambda_2$, and~$m$.

\end{proof}

\end{appendix}

\bibliographystyle{ecta} 
\bibliography{ref}		

\end{document}